\documentclass[twoside,11pt]{article}

\usepackage{amsmath}
\usepackage{amsthm}
\usepackage{amssymb}
\usepackage{bbm}
\usepackage[top=0.85in, bottom=0.85in, left=0.85in, right=0.85in]{geometry}
\newcommand {\bea}{\begin{eqnarray}}
	\newcommand {\eea}{\end{eqnarray}}

\newtheorem{theorem}{Theorem}[section]

\newtheorem{coro}{Corollary}[section]
\newtheorem{example}{Example}[section]

\newtheorem{remark}{Remark}[section]
\newtheorem{lemma}{Lemma}[section]
\newtheorem{definition}{Definition}[section]

\newtheorem{proposition}[lemma]{Proposition}

\newcommand{\md}{\mathbf{d}}

\usepackage{subcaption}
\usepackage{algorithm,algorithmic}

\def \beqs{\begin{eqnarray*}}
\def \enqs{\end{eqnarray*}}
\def \beq{\begin{eqnarray}}
\def \enq{\end{eqnarray}}
\usepackage{defs_general}
\usepackage{defs_notations}
\usepackage{footnote}
\date{}
\title{Markov $\alpha$-Potential Games}

\author{Xin Guo
\thanks{Department of Industrial Engineering \& Operations Research,
       University of California Berkeley, 
       Berkeley, CA 94720, USA \textbf{Email:} \{xinguo,  xinyu\_li\}@berkeley.edu}
\and
Xinyu Li
  \footnotemark[1]
  \thanks{Corresponding authors}
  \and
 Chinmay Maheshwari
  \thanks{Department of  Electrical Engineering and Computer Sciences, University of California Berkeley, Berkeley CA 94720 USA \textbf{Email:} \{chinmay\_maheshwari, shankar\_sastry\}@berkeley.edu}
  \footnotemark[2]
  \and
  Shankar Sastry 
  \footnotemark[3]
  \and
  Manxi Wu
  \thanks{
  Department of Civil and Environmental Engineering, University of California Berkeley,  Berkeley CA 94720 USA 
  \textbf{Email:} manxiwu@berkeley.edu
  }
}
\begin{document}
\maketitle

\begin{abstract}
We propose a new framework of Markov $\mnpg$-potential games to study Markov games. We show that any Markov game with finite-state and finite-action is a Markov $\alpha$-potential game, and establish the existence of an associated $\alpha$-potential function. Any optimizer of an $\alpha$-potential function is shown to be an $\alpha$-stationary Nash equilibrium. We study two important classes of practically significant Markov games, Markov congestion games and the perturbed Markov team games, via the framework of Markov $\alpha$-potential games, with  explicit characterization of an upper bound for $\mnpg$ and its relation to game parameters. Additionally, we provide a semi-infinite linear programming based formulation to obtain an upper bound for $\alpha$ for any Markov game. Furthermore, we study two equilibrium approximation algorithms, namely the projected gradient-ascent algorithm and the sequential maximum improvement algorithm, along with their Nash regret analysis, and corroborate the results with numerical experiments.    
\end{abstract}

\section{Introduction}
{Designing non-cooperative multi-agent systems interacting within a shared dynamic environment is a central challenge  in many existing and emerging autonomy applications, including autonomous driving, smart grid management, and e-commerce. Markov game, proposed in \cite{shapley1953stochastic}, provide a mathematical framework for studying such interactions \cite{zhang2021multi}. A primary objective in these systems is for agents to reach a \textit{Nash equilibrium}, where no agent benefits from changing its strategy unilaterally. However, designing algorithms for approximating or computing Nash equilibrium are generally intractable \cite{papadimitriou2007complexity}, unless certain structure of underlying multi-agent interactions are exploited. There is a rich line of literature on equilibrium computation and approximation algorithms for Nash equilibrium in Markov zero-sum games (see  
\cite{sayin2021decentralized} 
and references therein), Markov team games (see 
\cite{baudin2021best}
and references therein), symmetric Markov games (see \cite{yongacoglu2023satisficing}), and in particular, Markov potential games (see \cite{maheshwari2022independent, zhang2021gradient, leonardos2021global, narasimha2022multi} and references therein) and its generalization to weakly acyclic games (see \cite{arslan2016decentralized, yongacoglu2021decentralized} and references therein).}

 In this paper, we propose the \emph{Markov $\alpha$-potential game} framework, where changes in an agent's long-run utility from unilateral policy deviations are captured by an ``\emph{$\alpha$-potential function}'' and a parameter $\alpha$ (Definition \ref{def: mnpg}). We establish that any finite-state, finite-action Markov game is a Markov $\alpha$-potential game for some $\alpha \geq 0$, and there exists an $\alpha$-potential function (Theorem \ref{lemma: existence of feasible sol}). Furthermore, we show that any optimizer of an $\alpha$-potential function, if it exists, is an $\alpha$-stationary Nash equilibrium (Proposition \ref{prop:approximate_eq}).

{
Markov \(\alpha\)-potential games generalize the framework of Markov potential games (MPGs).
MPGs, originally proposed in \cite{macua2018learning} and \cite{leonardos2021global}, correspond to the special case of  $\alpha=0$ and extend a rich body of literature on static potential games (or static congestion games) \cite{monderer1996potential}.  The MPG structure has enabled learning algorithms with convergence guarantees to Nash equilibrium (e.g., \cite{maheshwari2022independent, ding2022independent}). However, two main challenges remain: (1) the lack of real-world examples that can be provably shown to be MPGs, and (2) the difficulty of certifying games as MPGs and constructing potential functions, except in special cases (e.g., state-independent transitions or identical payoffs \cite{narasimha2022multi, leonardos2021global}). Our $\alpha$-potential game framework addresses both challenges: it shows that any finite-state, finite-action Markov game is a Markov \(\alpha\)-potential game and provides a semi-infinite linear programming approach to certify MPGs (Section \ref{ssec: FindClosestMPG}).
}

Our Markov $\mnpg$-potential games framework extends the static near-potential games, proposed in \cite{candogan2013near, candogan2011flows}, to Markov games. Unlike static games, where the nearest potential function always exists, the existence of an $\alpha$-potential function requires additional analysis (Theorem \ref{lemma: existence of feasible sol}). Moreover, while finding the nearest static potential function involves finite-dimensional linear programming, computing the $\mnpg$ and its potential function requires solving a semi-infinite linear programming problem, as the $\alpha$-potential function spans both state and policy spaces, the latter being uncountable.
We derive explicit upper bounds on the parameter \(\alpha\) for two classes of relevant games. First, we consider \emph{Markov congestion games  (MCGs)}, where each stage game is a congestion game (proposed in \cite{rosenthal1973class})  and the state transition depends on agents' aggregate resource utilization. {{}This is equivalent to Markov games where each stage is a  static potential game, as static congestion games and static potential games are equivalent \cite{monderer1996potential}.} This class models applications like dynamic routing, communication networks, and robotic interactions \cite{cui2022learning, sun2023distributed, sun2024imagined, kavuncu2021potential}. We show that the upper bound on \(\alpha\) for MCGs scales linearly with the state and resource set sizes, and inversely with the number of agents (Proposition \ref{prop:mcg}). Second, we consider \emph{perturbed Markov team games (PMTGs)}, which generalize Markov team games by allowing utility deviations from the team objective. We provide an upper bound for PMTGs that scales with the magnitude of these deviations (Proposition \ref{prop: PerturbedMarkovTeamGame}). 
{{} For both MCGs and PMTGs, we calculate an upper bound on \(\alpha\) by using a \textit{specific} candidate \(\alpha\)-potential function to compute an analytical upper bound on \(\alpha\). However, this upper bound can be loose. In such cases, the semi-infinite linear programming method described in Section \ref{ssec: FindClosestMPG} can be used to obtain tighter numerical estimates of \(\alpha\).
}

 We propose two algorithms to approximate stationary Nash equilibrium in Markov $\alpha$-potential games. We study the Nash-regret of both algorithms and characterize its dependence on $\mnpg$ (Theorems \ref{thm: PG regret bound} and \ref{thm: SeqMaxImprovement}). First, we analyze the \emph{projected gradient-ascent algorithm} (Algorithm \ref{alg: PolicyGradient}), originally proposed in \cite{ding2022independent} for MPGs, in the context of Markov $\alpha$-potential games by bounding the path length of policy updates using changes in the $\alpha$-potential function and $\alpha$. {{}Following our proof technique, the analysis of many existing algorithms for MPGs can be extended similarly to Markov \(\alpha\)-potential games.} 
Second, we propose a \emph{new algorithm} called the \emph{sequential maximum improvement algorithm} (Algorithm \ref{alg: SequentialPureBRMaxPlayer}) and derive its Nash-regret. {{}The main technical novelty in the analysis is to bound the maximum improvement of a ``smoothed" Q-functions with respect to change in policies (aka ``path length of policies''), which in turn is bounded by cumulative change in \(\alpha\)-potential function (Lemma \ref{lemma:technical}). For $\alpha=0$, this algorithm and its analysis are independently relevant to MPGs.} We numerically validate these algorithms on examples of MCGs and PMTGs.

\subsection{Additional Related Works}
Our work on Markov \(\alpha\)-potential games is related to the literature on weakly acyclic Markov games, proposed in \cite{arslan2016decentralized}.
Weakly acyclic Markov games extend weakly acyclic static games to Markov games, encompassing MPGs as a special case. Unlike MPGs, weakly acyclic Markov games do not require the existence of an exact potential function, instead retain  many key properties of potential games, such as the existence of pure equilibria and finite strict best-response paths. Just as MPGs, most games are not weakly acyclic, and determining whether a game is weakly acyclic remains an open problem. 
On one hand, the introduction  of a Markov \(\alpha\)-potential games allows for design and analysis of algorithms as a game diverges from a MPG.  On the other hand, if a game is weakly acyclic,  it is an $\alpha$-potential game with the value of 
\(\alpha\) not necessarily zero.
Exploring the connection between these two approaches and how they might be used together to analyze general Markov games is an interesting and open direction for future research.

Our Algorithm \ref{alg: PolicyGradient} for Markov $\alpha$-potential games is connected with a substantial body of work on learning approximate Nash equilibria (NEs) in MPGs (see \cite{maheshwari2022independent, ding2022independent, mao2021decentralized, song2021can, fox2022independent, sun2023provably, zhangglobal}). The first global convergence result for the policy gradient method in MPGs was established in \cite{leonardos2021global}. Additionally, these algorithms have been studied in both discounted infinite horizon settings \cite{ding2022independent, fox2022independent} and finite horizon episodic settings \cite{mao2021decentralized, song2021can}. Other methods, such as natural policy gradient \cite{fox2022independent, sun2023provably, zhangglobal} and best-response based methods \cite{maheshwari2022independent}, have also been explored.


Our Algorithm \ref{alg: SequentialPureBRMaxPlayer} is reminiscent of the ``Nash-CA'' algorithm developed for MPGs in \cite{song2021can}, which 
requires each player to {sequentially} compute the best response policy using an RL algorithm in each iteration; in contrast,  our algorithm  only computes a smoothed \emph{one-step} optimal deviation. One-step optimal deviation based algorithms has also been studied for MPGs \cite{maheshwari2022independent, cartea2022algorithms}. {{} Additionally, incorporating  smoothness for better performance in Markov games is also studied in \cite{cen2021fast, evans2024learning, mertikopoulos2016learning}. }

Finally, a recent work \cite{cui2022learning} introduces an approximation algorithm for MCGs and investigates the Nash-regret. Their  results and approach  are tailored exclusively for congestion games, whereas our work focuses on a  broader framework of Markov \(\mnpg\)-potential games.

\subsection{Notations} 
For any \(n\in \mathbb{N}\),  \([n]:=\{1,2,3,...,n\}\).
For a finite set $X$, $\mathcal{P}(X)$ denotes the set of probability distributions over $X.$ For any function $f: X \ra \mathbb{R}$,  the $L_\infty$-norm is defined by $\|f\|_{\infty} = \max_{x\in X}|f(x)|$, the $L_1$-norm is $\|f\|_{1}=\sum_{x\in X}|f(x)|$,  and the $L_2$-norm is $\|f\|=\sqrt{\sum_{x\in X}|f(x)|^2}.$

\section{Markov Games}\label{sec: Setup}
Consider an $N$-player Markov game $\mathcal{G}$ as characterized by the tuple $\langle \playerSet, \stateSet, (\actSet_i)_{i\in\playerSet}, (\stagePayoff_i)_{i\in\playerSet}, \transition, \discount\rangle$, where 
$\playerSet = (1,2, \cdots N)$ is the finite set of $N\in \mathbb{N}$ players, 
$S$ is the finite set of states, \(A_i\) is the finite set of actions of player \(i\in \playerSet\) and \(A := \times_{i\in \playerSet}A_i\) is the set of joint actions of all players,  \(\stagePayoff_i: S \times A \to \mathbb{R}\) is the one-stage payoff function of player \(i\in \playerSet\), $\transition = (\transition(s'|s, a))_{s, s' \in S, a \in A}$ is the probability transition kernel such that $\transition(s'|s, a)$ is the probability of transitioning to state $s'\in S$ given the current state $s\in S$ and action profile $a\in A$, and
$\discount \in [0,1)$ is the discount factor. 

The game proceeds in discrete time steps.  At each time step $k=0,1,2,\cdots$, the state of the game is $s^k\in S,$ the action taken by player \(i\in \playerSet\) is $a_i^k \in A_i$, and the joint action of all players is $a^k = (a_i^k)_{i\in \playerSet} \in A$. Once players select their  actions, each player $i \in \playerSet$ observes her one-stage payoff  $\stagePayoff_i(s^k, a^k) \in  \mathbb{R}$, and the system transits to state $s^{k+1}$, where  $s^{k+1} \sim P(\cdot | s^k, a^k)$.
In this study, we assume that the action taken by any player is based on a randomized stationary Markov policy, as in the Markov games literature \cite{daskalakis2020independent, leonardos2021global, ding2022independent}. That is,  for any player \(i\in \playerSet\), the action selected at time step \(k\) is \(a_i^k\sim\pi_i(\cdot|s^k)\), and 
the joint policy of all players is \(\policy = (\policy_i)_{i\in \playerSet} \in \Pi:= \times_{i \in \playerSet} \Pi_i\),
with  $\Pi_i \coloneqq\{\pi_i: S\rightarrow \mathcal{P}(A_i) \}$. The joint policy of all players except player \(i\) is denoted as \(\policy_{-i} = (\policy_j)_{j \in \playerSet \setminus \{i\}} \in \Pi_{-i} \coloneqq \times_{j \in \playerSet \setminus \{i\}} \Pi_j\). Given $\pi \in \Pi$, the probability of the system transiting  from $s$ to $s'$ is denoted as \(
    P^{\pi}(s'|s) \coloneqq 
    \mathbb{E}_{a\sim \pi} [P(s'|s,a)]
    \).
The accumulated reward (a.k.a. the \emph{utility function}) for player $i$, given the initial state  $s \in S$ and the joint policy $\pi\in \Pi$, is 
\begin{equation} \label{eqn: value_function}
    V_i({s}, \pi):=\mathbb{E}_{\pi}\left[\sum_{k=0}^{\infty} \discount^k \stagePayoff_i\left(s^k, a^k\right) \mid {s}^0={s}\right],
\end{equation}
where {$\discount \in [0,1)$} is the discount factor, $a^k \sim \pi\left(s^k\right)$, and $s^{k+1} \sim P\left(\cdot | s^k, a^k\right)$.
 Denote also \(V_i(\mu, \pi) \coloneqq \mathbb{E}_{s \sim \mu} [V_i (s, \pi)]\), if the initial state follows a distribution $\mu \in \mathcal{P}(S)$.
Additionally, define the discounted state visitation distribution as 
    \(d_\mu^{\pi}(s):=(1-\discount)\sum_{t=0}^{\infty}\discount^kP(s^k = s|s^0\sim\mu)\).
To analyze this game, we adopt the solution concept of  $\epsilon$-stationary Nash equilibrium (NE).
\begin{definition}\label{def: eps NE} \textit{($\epsilon$-stationary Nash equilibrium).}
    For any $\epsilon \geq 0$, a policy profile \(\policyNash  = (\pi_i^*, \pi_{-i}^*)\) is an \emph{$\epsilon$-stationary Nash equilibrium} of the Markov game $\game$ if for any \(i\in \playerSet\), any \(\policy_i\in\Pi_i\), and any \(\initDist\in\mathcal{P}(\stateSet)\),
    \(
        \vFunc_i(\mu, \policyNash_i,\policyNash_{-i}) \geq \vFunc_i(\mu, \policy_i,\policyNash_{-i}) -\epsilon.
    \)
    \end{definition} 

When $\epsilon=0$, it is simply called a \emph{stationary NE}, which always exists in our setup \cite{fudenberg1998theory}.

\section{Markov \texorpdfstring{$\mnpg$}{a}-potential games}\label{sec: NearMPG}
In this section, we introduce the framework of Markov $\mnpg$-potential games. We show that any Markov game can be analyzed under this framework.
First, we introduce some preliminaries. We define a metric \(\md\) on \(\Pi\) as follows: for any \(\pi, \tilde\pi \in \Pi\),
\begin{align}\label{eqn: def_of_metric_d}
        \md_i\left(\pi_i, \tilde{\pi}_i\right) \coloneqq &\max_{s \in S, a_i \in A_i } \left|\pi_i\left(a_i \mid s\right)-\tilde{\pi}_i\left(a_i \mid s\right)\right|, \quad \forall i\in \playerSet, \notag \\
        \md ({\pi}, \tilde{{\pi}})\coloneqq &\max _{i \in \playerSet} \md_i\left(\pi_i, \tilde{\pi}_i\right).
\end{align}
Evidently, the sets of policies $\{\Pi_i\}_{i\in \playerSet}$ are compact in the topology induced by the metrics $\left\{\md_i\right\}_{i \in \playerSet}$,  $\Pi$ is compact in the topology induced by $\md$, and  the utility functions are continuous with respect to $\pi$ under the metric $\md$ \cite{yongacoglu2023satisficing}. 
Next, we introduce the notion of maximum pairwise distance between a Markov game and a real-valued function defined on \(S\times \Pi\). 
\begin{definition}\textit{(Maximum pairwise distance).} \label{def: MaxPairDist}
    Given any Markov game $\mathcal{G}$ and a  function $\Psi: S\times\Pi\ra\R,$  the \emph{maximum pairwise distance} $\widehat{\md}$ between $\Psi$ and $\mathcal{G}$ is defined as 
    $$
    \begin{aligned}
         \widehat{\md}(\Psi,\mathcal{G}) \coloneqq
     &\sup_{\substack{
                  s\in S, i\in \playerSet, \\
                \pi_i,\pi'_i \in \Pi_i, 
                \pi_{-i} \in \Pi_{-i}}
            }
            \Big| \Psi\left(s, \pi_i^{\prime}, \pi_{-i}\right)-\Psi\left(s, \pi_i, \pi_{-i}\right) - \left(V_i\left(s, \pi_i^{\prime}, \pi_{-i}\right)-V_i\left(s, \pi_i, \pi_{-i}\right)\right) \Big|.
    \end{aligned}
    $$
\end{definition}

Definition \ref{def: MaxPairDist} generalizes the concept of maximum pairwise distance from \cite[Definition 2.3]{candogan2013near}, extending it from static games (action profiles) to Markov games, where the distance is measured over policies that map states to action distributions.
Next, we introduce the notion of a game elasticity parameter, which is useful for defining Markov \(\alpha\)-potential games. Intuitively, this parameter captures the smallest value of the maximum pairwise distance between any function in a set \(\mathcal{F}_{\mathcal{G}}\) (to be defined shortly) and \(\mathcal{G}\).

\begin{definition}\label{def: alpha_param}\textit{(Game elasticity parameter).}
Given any game $\mathcal{G}$, its \emph{game elasticity parameter} $\alpha$ is defined as
\begin{align}\label{eqn: def eps MPG}
\alpha := \inf_{\Psi \in \mathcal{F^\game}} 
            \widehat{\md}(\Psi,\mathcal{G}),
    \end{align}    
    where \(\mathcal{F}^\game := \{\Psi:S\times \Pi \rightarrow \mathbb{R} \ \text{s.t.} \ \| \Psi\|_{\infty} \leq \frac{2N}{1-\gamma}\max_{i\in\playerSet}\| u_i\|_{\infty}\}\) 
  is a class of bounded uniformly equi-continuous function on \(\Pi\). \footnote{A set $\mathcal{F}$ of functions $f: S\times \Pi \ra \R$
    is called \emph{uniformly equi-continuous on $\Pi$}, if there exists \(\delta_{\mathcal{F}}: \mathbb{R}_+\rightarrow\mathbb{R}_+\) such that for every \(\epsilon>0\), \(|f(s, \pi)-f(s, \pi')| \leq \epsilon\) for all \(f\in \mathcal{F},s\in S\), \( \pi,\pi'\in \Pi\) such that \({\md}(\pi,\pi')\leq \delta_{\mathcal{F}}(\epsilon)\).}
\end{definition}
Our choice of the specific value of the upper bound on functions in \(\mathcal{F}^{\mathcal{G}}\) is 
useful for the proof of Proposition 4.1.

Clearly $\alpha<\infty$ as one can take $\Psi=0$ in \eqref{eqn: def eps MPG} to ensure  
$\alpha \leq 2\|V_i \|_{\infty} {<} \infty$.

{{} Furthermore, the game elasticity parameter depends on variety of game parameters, including the number of players, the action and state sets, the utility function values, the Markov state transition dynamics, and the discount factor.}

Next, we define  Markov \(\alpha\)-potential games.
\begin{definition}\label{def: mnpg}\textit{(Markov $\alpha$-potential game).}
A Markov game \(\game\) is a \emph{Markov $\alpha$-potential game} if \(\alpha\) is the game elasticity parameter. Furthermore, any \(\Phi\in\mathcal{F}^\game\) such that 
$\widehat{\md}(\Phi,\game) = \alpha$   is called an \emph{$\alpha$-potential function} of $\mathcal{G}$.
\end{definition}

Next, we present a useful property due to Definition \ref{def: mnpg}.
{
\begin{coro}\label{coro:V_minus_V_minus_Phi_minus_Phi} Let $\game$ be a Markov $\alpha$-potential game with $\alpha$-potential function $\Phi$. Then, for any \(s\in S, i\in \playerSet, \pi_i, \pi'_i \in \Pi_i, \pi_{-i}\in\Pi_{-i}\),
\begin{equation}\label{eqn: V_minus_V_minus_Phi_minus_Phi}
    \begin{aligned}
        |&V_i(s,\pi_i,\pi_{-i}) -V_i(s,\pi'_i,\pi_{-i}) - (\Phi(s,\pi_i,\pi_{-i})- \Phi(s,\pi'_i,\pi_{-i}))| \leq \alpha.
    \end{aligned}
\end{equation}
\end{coro}
}

Next, we show the existence of an $\alpha$-potential function.
\begin{theorem}\label{lemma: existence of feasible sol}\textit{(Existence of $\alpha$-potential function).}
For any Markov game $\mathcal{G}$, 
there exists $\Phi \in \mathcal{F}^\game$ such that 
\(
\widehat{\md}(\Phi,\mathcal{G})
             =
             \inf_{\Psi \in \mathcal{F}^\game} \widehat{\md} (\Psi, \mathcal{G}).\)
\end{theorem}
\begin{proof}
Define a mapping $ \mathcal{F}^\game \times \Pi \times \Pi \ni (\Psi,\pi,\pi') \mapsto  h(\Psi,\pi,\pi') :=  \underset{{s\in S, i\in \playerSet}}{\max} 
             \big| \Psi\left(s, \pi_i^{\prime}, \pi_{-i}\right)-\Psi\left(s, \pi_i, \pi_{-i}\right) - \left(V_i\left(s, \pi_i^{\prime}, \pi_{-i}\right)-V_i\left(s, \pi_i, \pi_{-i}\right)\right) \big| \in  \R$. Note that such $h$ is continuous under the standard topology induced by sup-norm on $\mathcal{F}^\game \times \Pi \times \Pi.$ 
By Berge's maximum theorem, \(g(\Psi)\coloneqq {\max}_{{\pi,\pi'\in \Pi}} h(\Psi, \pi,\pi')\) is continuous with respect to $\Psi$.
Since \(\mathcal{F}^\game \) is uniformly bounded and uniformly equi-continuous, the Arzelà–Ascoli theorem implies that $\mathcal{F}^\game$ is relatively compact in $\mathcal{C}^\Pi$, where $\mathcal{C}^\Pi \coloneqq  \{f: S\times \Pi\rightarrow \mathbb{R} ~| ~\forall  s\in S,    f(s,\cdot) \text{ is a continuous function}\}$  \cite{rudin1953principles}.
Finally, by the extreme-value theorem \cite{rudin1953principles},
there exists a function $\Phi \in \mathcal{F}^\game$ such that \(
\widehat{\md}(\Phi,\mathcal{G})
             =
             \inf_{\Psi \in \mathcal{F}^\game} \widehat{\md} (\Psi, \mathcal{G}).\)
\end{proof}

{{}Corollary \ref{coro:V_minus_V_minus_Phi_minus_Phi} and Theorem \ref{lemma: existence of feasible sol} jointly show that for any Markov game \(\game,\) an $\alpha$-potential function exists such that the gap between the change in the utility function of any agent due to a unilateral change in its policy and the change in \(\alpha\)-potential function is at most \(\alpha\). 
Next, we show that any optimizer of $\alpha$-potential function with respect to policy $\pi$ yields an $\alpha$-Nash equilibrium (NE) of game $\mathcal{G}$.}
\begin{proposition}\label{prop:approximate_eq} 
Given a Markov \(\mnpg\)-potential game $\game$ with an {$\mnpg$-potential function} \(\Phi\), for any $\epsilon > 0$, if there exists a \(\pi^\ast\in \Pi\) such that for every 
$s\in S$,
$
\Phi(s, \pi^*) + \epsilon \geq \sup_{\pi \in \Pi} \Phi(s, \pi),  
$ then
 \(\pi^\ast\in \Pi\) is an $(\mnpg + \epsilon)$-stationary NE of \(\game\). 
\end{proposition}

\begin{remark}\label{rmk:NE_optimizer}
Note that Proposition \ref{prop:approximate_eq}  holds  for any function $\Psi{\in \mathcal{F}^\game}$ that yields an upper bound for $\alpha$. That is, given a Markov $\alpha$-potential game $\game$ and a function $\Psi$ satisfying 
\begin{equation*}
    \begin{aligned}
        |&V_i(s,\pi_i,\pi_{-i}) -V_i(s,\pi'_i,\pi_{-i}) - (\Psi(s,\pi_i,\pi_{-i})- \Psi(s,\pi'_i,\pi_{-i}))| \leq \bar\alpha, \quad \forall s\in S, \pi_{i},\pi_{i}'\in \Pi_i, \pi_{-i}\in \Pi_{-i},
    \end{aligned}
\end{equation*}
for some $\bar\alpha \in [\alpha,\infty)$,
then for any $\pi^* \in \Pi$ such that for every \(s\in S\),  \(
\Psi(s,\pi^*) + \epsilon \geq \sup_{\pi\in\Pi} \Psi(s,\pi),
\)
\(\pi^*\) is an \((\bar\alpha + \epsilon)\)-stationary NE of \(\game\).
\end{remark}

\section{Examples of Markov \texorpdfstring{$\mnpg$}{a}-potential game}\label{ssec: ExamplesNearMPG}
In this section, we present three important classes of games, Markov potential games, Markov congestion games, and perturbed Markov team games, which can be analyzed through the framework of Markov $\mnpg$-potential games.

\subsection{Markov potential game}
A game is a Markov potential game if there exists {an auxiliary function  (a.k.a. potential function)} such that when a player unilaterally deviates from her policy,  the change of the potential function is equal to the change of her utility function.  
    \begin{definition}[Markov potential games \cite{leonardos2021global}]\label{def: PotentialGame} A Markov game \(\game\) is a \emph{Markov potential game} (MPG) if there exists a potential function 
    \(\potential: S \times \Pi\ra \mathbb{R}\) 
    such that for any \(i \in \playerSet\),  $s\in S$,  \(\policy_i,\policy_i'\in \Pi_i\), and  \(\policy_{-i}\in \Pi_{-i}\),
\(
        \Phi(s, \policy_i',\policy_{-i}) -  \Phi(s, \policy_i,\policy_{-i})= \vFunc_i(s, \policy_i',\policy_{-i}) -  \vFunc_i(s, \policy_i,\policy_{-i}).\)
    \end{definition}

\begin{proposition}\label{prop: exampleMPG}
    An MPG is a Markov $\alpha$-potential game with $\alpha=0.$
\end{proposition}

\subsection{Markov congestion game}
The Markov congestion game (MCG) $\game_{\textsf{mcg}}$ is a dynamic counterpart to the static congestion game introduced by \cite{monderer1996potential}, involving a finite number of players using a finite set of resources. Each stage of $\game_{\textsf{mcg}}$ is a static congestion game with a state-dependent reward function for each resource, and the state transition depends on the aggregated usage of each resource by the players. 
Specifically, let the finite set of resources in the one-stage congestion game be denoted as $E$. The action $a_i \in A_i \subseteq 2^{E}$ of each player $i \in \playerSet$ represents the set of resources chosen by player $i$. Here, the action set $A_i$ is the set of all resource combinations that are feasible for player $i$. The total usage demand of all players is 1, and each player's demand is assumed to be $1/N$.

Given an action profile \(a = (a_i)_{i \in \playerSet}\), the aggregated usage demand of each resource \(e \in E\) is given by
\begin{align}\label{eq:w}
w_e(a) =  \frac{1}{N}\sum_{i\in \playerSet}\mathbbm{1}(e \in a_i) .
\end{align}
In each state \(s\), the reward for using resource \(e\) is denoted as $(1/N)\cdot c_e(s, w_e(a))$. Thus, the one-stage payoff for player \(i \in \playerSet\) in state \(s \in S\), given the joint action profile \(a \in A\), is
 \(u_i(s,a) = (1/N)\cdot\sum_{e\in a_i}c_e(s, w_e(a))\).
The state transition probability, denoted as \(P(s'|s,w)\), depends on the aggregate usage vector \(w = (w_e)_{e \in E}\), which is induced by the players' action profile as in \eqref{eq:w}. The set of all feasible aggregate usage demands is denoted by \(W\).

The next proposition shows that, under a regularity condition on the state transition probability,   \(\game_{\textsf{mcg}}\) is a Markov \(\mnpg\)-potential game such that the upper bound of $\mnpg$  scales linearly with respect to the Lipschitz constant $\zeta$, the size of state space $|S|$, resource set $|E|$, and decreases as $N$ increases. 
\begin{proposition}\label{prop:mcg}
If there exists some $\zeta>0$ such that for any $s, s'\in S, w, w'\in W,$ 
\(
    |P(s'|s,w)-P(s'|s,w')| \leq \zeta  \|w-w'\|_{1}, \)
then the
    congestion game \(\game_\textsf{mcg}\) is a Markov \(\mnpg\)-potential game with 
   {\(
        \mnpg \leq {2\zeta\discount |S| |E|\sup_{s,\pi}\Psi(s,\pi)}/({N(1-\discount)})
    \)},
where\begin{align}\label{eq:potential}
\Psi(\mu, \pi):=\frac{1}{N}\avg_{\mu,\pi}\Bigg[\sum_{k=0}^{\infty}\discount^k \bigg(\sum_{e\in E}\sum_{j=1}^{w_e^k N}c_e\Big(s^k, \frac{j}{N}\Big)\bigg)\Bigg],
\end{align}
such that \(s^0\sim \initDist\), the aggregate usage vector $w^k = (w_e^k)_{e \in E}$ is induced by \(a^k\sim \policy(s^k)\), and \(s^k\sim \transition(\cdot|s^{k-1},w^{k-1})\).
\end{proposition}


\subsection{Perturbed Markov team game}
A Markov game is called a perturbed Markov team game (PMTG) \(\pertCommonGame\) if the payoff function for each player \(i \in \playerSet\) can be decomposed as \(u_i(s, a) = r(s, a) + \xi_i(s, a)\). Here, \(r(s, a)\) represents the common interest of the team, and \(\xi_i(s, a)\) represents player \(i\)'s heterogeneous preference, such that \(\|\xi_i\|_{L_\infty} \leq \kappa\), where \(\kappa \geq 0\) measures each player's deviation from the team's common interest. As \(\kappa \to 0\), \(\game_{\textsf{pmtg}}\) becomes a Markov team game, which is an MPG \cite{leonardos2021global}.

The next proposition shows that a \(\game_{\textsf{pmtg}}\) is a Markov \(\mnpg\)-potential game, and the upper bound of \(\mnpg\) decreases as the magnitude of the payoff perturbation \(\kappa\) decreases.

\begin{proposition}\label{prop: PerturbedMarkovTeamGame}
    A perturbed Markov team game \(\pertCommonGame\) is a Markov \(\mnpg\)-potential game with 
{\(
        \mnpg \leq \frac{2\kappa}{(1-\discount)^2}.
    \)}
\end{proposition}

\section{Finding an upper bound of \texorpdfstring{$\mnpg$}{a}}
\label{ssec: FindClosestMPG}
{{} The analysis of MCG and PMTG in Section \ref{ssec: ExamplesNearMPG} utilizes a specific form of the Markov \(\alpha\)-potential function to obtain an upper bound on \(\alpha\). In this section, we provide an optimization-based procedure to find an upper bound on \(\alpha\) by also computing the \(\alpha\)-potential function.}

Our approach is based on changing the feasible set of the optimization problem in \eqref{eqn: def eps MPG} to \(\tilde{\mathcal{F}}^\game\), defined as follows:
\begin{align}\label{eqn: F tilde}
    \Tilde{\mathcal{F}}^\game := 
    &\Bigg\{ \Psi(s,\pi)
    = \sum_{{s'\in S, a'\in A}}d^{s}(s',a';\pi) \phi(s',a'),  \forall s\in S,
\pi \in \Pi \bigg|
    \phi: S\times A \rightarrow\R \\
    &\qquad \text{ s.t. } 
    \| \phi\|_{\infty} \leq N \max_{i\in\playerSet} \| u_i\|_{\infty}\Bigg\},
\end{align}
where, for any $s\in S$,
$d^s(\cdot; \pi): S\times A \rightarrow \mathbb{R}$ is the \emph{state-action occupancy measure} induced due to \(\pi\), defined as follows:
$$
    d^{s}(s',a';\pi)   \coloneqq \pi(a'|s')\mathbb{E}_\pi\left[\sum_{k=0}^{\infty}\discount^k \mathbbm{1}(s^k =s') \Big|s^0 = s \right],
$$
where $a^k \sim \pi\left(s^k\right)$, and $s^{k+1} \sim P\left(\cdot | s^k, a^k\right)$.
{{} Intuitively, for any \(\Psi \in \tilde{\mathcal{F}}^{\mathcal{G}}\), there exists \(\phi: S \times A \to \mathbb{R}\) such that \(\Psi(s, \pi)\) represents the long-horizon discounted value of a Markov decision process with state transition \(P\), starting from state \(s\), using policy \(\pi \in \Pi\), and one-step utility \(\phi\).}

\begin{proposition}\label{lemma: Phi_constructed_MDP}
For any Markov $\alpha$-potential game $\game$, $\Tilde{\mathcal{F}}^\game \subseteq \mathcal{F}^\game.$ That is, $\bar\alpha \geq \alpha$ with
\begin{equation}\label{eqn: relaxed_alpha_bar}
    \bar{\alpha}\coloneqq\inf_{\Psi \in \Tilde{\mathcal{F}}^\game} 
            \widehat{\md}(\Psi,\mathcal{G}).
\end{equation}
\end{proposition}
{Using Remark \ref{rmk:NE_optimizer}, we can conclude that any optimizer of \(\bar{\Psi}\), where \(\widehat{\md}(\bar{\Psi}, \mathcal{G}) = \bar{\alpha}\), can be used to find a \(\bar{\alpha}\)-stationary NE for the game \(\mathcal{G}\).
}

Next, we provide an optimization based method to compute \(\bar{\alpha}\). Note that \eqref{eqn: relaxed_alpha_bar} can be reformulated as follows:
\begin{align}
    \min_{\substack{ y\in \R \\
    \phi:S\times A\ra \R}}\quad  &  y\label{eq: LinearSemiInfiniteFormulation}\\ 
    \text{s.t.}\quad &  \Big|  \sum_{s',a'}(d^{s}(s',a';\pi_i,\pi_{-i}) - d^{s}(s',a';\pi_i',\pi_{-i}))  \cdot (\phi-u_i)(s',a')\Big| \leq y, \tag{C1} \\ &\forall s\in S, ~\forall i \in \playerSet, ~ \forall \pi_i, \pi'_i \in {\Pi}_i, ~ \forall \pi_{-i}\in {\Pi}_{-i}, \notag
    \\ & |\phi(s,a) | \leq  N \max_{i\in\playerSet} \| u_i\|_{\infty},\quad  \forall s\in S, a\in A.\notag
\end{align}
Here, we use \(    
    V_i(s,\pi)= \sum_{s'\in S, a'\in A}d^{s}(s',a';\pi) u_i(s',a')
\) and \(\Psi(s,\pi)= \sum_{s'\in S, a'\in A}d^{s}(s',a';\pi) \phi(s',a')\), for some \(\phi: S\times A \rightarrow \mathbb{R}\). 
Note that 
\eqref{eq: LinearSemiInfiniteFormulation} is a semi-infinite linear programming where the objective is a linear function with an uncountable number of linear constraints. Particularly, in (C1) there is one linear constraint corresponding to each tuple \((s, i,\pi_i,\pi_i',\pi_{-i})\). Moreover, the coefficients of each linear constraint in (C1) are composed of state-action occupancy measures which are computed by solving a Bellman equation. There are a number of algorithmic approaches to solve semi-infinite linear programming problems 
\cite{tadic2003randomized,hettich1993semi}. 
In Appendix \ref{sec: FindNearPotential}, we adopt an algorithm from \cite{tadic2003randomized} to solve \eqref{eq: LinearSemiInfiniteFormulation} and find (an upper bound of) \(\alpha\). Figure \ref{fig:gamma_var} illustrates how \(\alpha\) varies with different discount factors \(\gamma\) in a PMTG. Note that the growth of the numerical estimate of \(\alpha\) is much more benign than the analytical characterization obtained in Proposition \ref{prop: PerturbedMarkovTeamGame}. 

\begin{figure}
        \centering
\includegraphics[width= 0.4\linewidth]{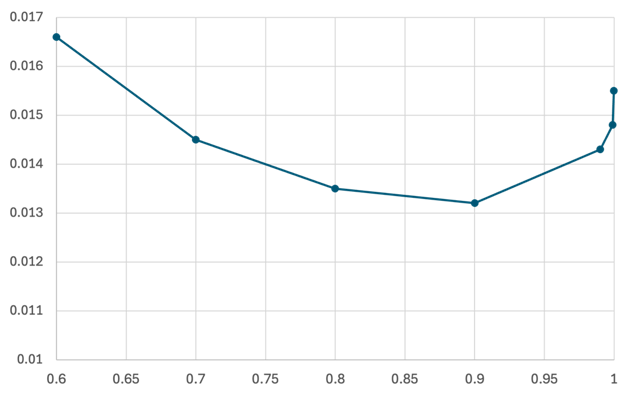}
        \caption{Variation of \(\alpha\) with the discount factor in the perturbed Markov team game with \(N=3\) and perturbation parameter \(\kappa = 0.1\). The setup of this game is same as that in Section \ref{sec: numerics} with \(\lambda_1=\lambda_3=0.8,\lambda_2=\lambda_4=0.2\).}
        \label{fig:gamma_var}
    \end{figure}

\section{Approximation algorithms and Nash-regret analysis}\label{sec: Algorithms}
In this section, we present two equilibrium approximation algorithms for Markov \(\alpha\)-potential games: the \textit{projected gradient-ascent algorithm}, proposed in \cite{ding2022independent} for MPGs, and the \textit{sequential maximum improvement algorithm}, where each player's strategy is updated based on a one-stage smoothed best response. We also derive non-asymptotic convergence rates for these algorithms in terms of Nash-regret, defined as 
\(
\text{Nash-regret}(T) := \frac{1}{T} \sum_{t=1}^T \max_{i \in \playerSet} R_i^{(t)},
\)
where 
\(
R_i^{(t)} := \max_{\pi_i^{\prime} \in \Pi_i} V_i\left( \mu, \pi_i^{\prime}, \pi_{-i}^{(t)} \right) - V_i\left( \mu, \pi^{(t)} \right),
\)
and \(\pi^{(t)}\) denotes the \(t\)-th iterate. Note that Nash-regret is always non-negative; if \(\text{Nash-regret}(T) \leq \epsilon\) for some \(\epsilon > 0\), then there exists \(t^\dagger\) such that \(\pi^{(t^\dagger)}\) is an \(\epsilon\)-stationary NE.

\subsection{Projected gradient-ascent algorithm} \label{sssec: AlgPG}
First, we define some useful notations. Given a joint policy \(\pi\in \Pi\), define player $i$'s \emph{$Q$-function} as 
\(
        Q_i^\pi(s,a_i)= \mathbb{E}_{a_{-i}\sim \pi_{-i}(s)}  \big[ u_i(s,a_i,a_{-i})
 + \discount \sum_{s'\in S}P(s'|s,a_i, a_{-i}) V_i(s',\pi)\big],
\)
   and denote \(\bQ_i^\pi(s) = (Q_i^\pi(s,a_i))_{a_i\in A_i}\).  
Let $\kappa_\mu$ denote the maximum distribution mismatch of $\pi$ relative to $\mu$, and let $\tilde{\kappa}_\mu$ denote the minimax value of the distribution mismatch of $\pi$ relative to $\mu$. That is,
     \begin{equation}\label{eqn: kappa mu defs}
         \kappa_\mu := \sup _{\pi \in \Pi}\left\|d_\mu^\pi / \mu\right\|_{\infty}, \quad \tilde{\kappa}_\mu := \inf_{\nu \in \mathcal{P}({S})} \sup_{\pi\in \Pi} \|d_\mu^\pi/\nu \|_\infty,
     \end{equation}
      where the division $d_\mu^\pi/\nu$ is evaluated in a component-wise manner.
      The algorithm iterates for \(T\) steps. We abuse the notation to use $Q_i^{(t)}$ to denote $Q_i^{\pi^{(t)}}$, and $\bQ_i^{(t)}$ to denote $\bQ_i^{\pi^{(t)}}$. In every step \(t\in [T-1]\), each player \(i\in \playerSet\) updates her policy following a projected gradient-ascent algorithm as in \eqref{eqn: policy grad Q update}. 
      
\begin{algorithm}\caption{Projected Gradient-Ascent Algorithm}
    \begin{algorithmic}
        \STATE{\textbf{Input:} Step size \(\eta\), for every \(i\in \playerSet, a_i\in A_i, s\in S\), set \(\pi_i^{(0)}(a_i|s)=1/|A_i|\).
}
\FOR{\(t = 0,1,2,...,T-1\)}
\STATE{For every \(i\in \playerSet, s\in S\), update the policies as follows}
\begin{equation}\label{eqn: policy grad Q update}
      \pi_i^{(t+1)}(s)= 
      {\texttt{Proj}}_{\Pi_i}\left(  \pi_i^{(t)}(s) + \eta {\bQ}_i^{(t)}(s) \right), 
\end{equation}
where 
\(\texttt{Proj}_{\Pi_i}\) denotes the orthogonal projection on \(\Pi_i\).
\ENDFOR
\end{algorithmic}\label{alg: PolicyGradient}
\end{algorithm}
{{}
\begin{remark}
    Algorithm 1 is not the standard policy gradient algorithm. The standard policy gradient is given by \(
\frac{\partial V_i^\pi(\rho)}{\partial \pi_i\left(a_i \mid s\right)} = 1/({1-\gamma})\cdot d_\rho^\pi(s) {Q}_i^\pi\left(s, a_i\right)\) \cite{ leonardos2021global}.
The RHS in the this equation scales with the state visitation frequency $d_\rho^\pi(s)$, which results in slow learning rate for states with low visitation frequencies under the current policy. To address this issue, \cite{ding2022independent} proposed to remove the term $d_\rho^\pi(s) /(1-\gamma)$ from the standard policy gradient update, which accelerates the learning for states with low visitation probabilities. We adopted the convention of \cite{ding2022independent} to call it ``policy gradient-ascent algorithm".
\end{remark}
}

\begin{theorem}\label{thm: PG regret bound}
    Given a Markov \(\mnpg\)-potential game with an \(\mnpg\)-potential function \(\Phi\) and an initial state distribution \(\mu\),
    the policy updates generated from Algorithm \ref{alg: PolicyGradient} satisfies
    \begin{itemize}
        \item[(i)] $\text{Nash-regret}(T)\leq  \mathcal{O}\lr{\frac{ \sqrt{\tilde{\kappa}_\mu \bar{A} N} }{(1-\discount)^{\frac{9}{4}} }\left( \frac{C_\Phi }{T}+ N^2 \mnpg \right)^\frac{1}{4}}$  with $\eta= \frac{(1-\discount)^{2.5} \sqrt{C_\Phi + N^2 \mnpg T}}{2 N \bar{A} \sqrt{T}}$;
        \item[(ii)] $\text{Nash-regret}(T) \leq  \mathcal{O}\left(\sqrt{\frac{ \min{(\kappa_\mu, |S|)}^4 N \bar{A}  }{(1-\discount)^6 }} \left(\frac{C_{\Phi}}{T}+ N^2 \alpha \right)^{\frac{1}{2}}\right)$ with $ \eta = \frac{(1-\discount)^4}{8 \min{( \kappa_\mu, |S| )}^3 N \bar{A}}$,
    \end{itemize}
where  \(\bar{A} \coloneqq \max_{i\in \playerSet }|A_i|\),
 $\kappa_\mu$ and $\tilde{\kappa}_\mu$ are defined in \eqref{eqn: kappa mu defs},
and $C_\Phi >0$ is a constant satisfying $\left|\Phi(\mu,\pi)-\Phi(\mu,\pi^{\prime})\right| \leq C_{\Phi}$ for any $\pi, \pi^{\prime} \in \Pi, \mu \in \mathcal{P}(S).$ 
\end{theorem}

We emphasize that the Nash-regret bounds in Theorem \ref{thm: PG regret bound} (also Theorem \ref{thm: SeqMaxImprovement} in the next section)
will hold even without knowing the exact form of $\Phi$ and the game elasticity parameter $\alpha.$
It is sufficient to have an upper bound $\bar{\alpha}$ for $\mnpg$ and an associated function $\Psi$ for which this upper bound holds.
In the special case of $\alpha=0$,
     the Nash-regret bound in Theorem \ref{thm: PG regret bound} recovers the Nash-regret bound from  \cite{ding2022independent} for MPG.

The proof of Theorem \ref{thm: PG regret bound} is inspired by \cite{ding2022independent} for the Nash-regret analysis of MPGs. First, we state multi-player performance difference lemma (Lemma \ref{lemma: performance diff}), which enables bounding the Nash-regret of an algorithm by summing the norms of policy updates, denoted as \(\|\pi_i^{(t+1)} - \pi_i^{(t)} \|\). The main modification for our analysis is to bound the sum of these policy update differences by the game elasticity parameter \(\alpha\) and the change in the \(\alpha\)-potential function \(\Phi\) (Lemma \ref{lemma: policy improve}).

\begin{lemma}[Performance difference (Lemma 1 in \cite{ding2022independent})]\label{lemma: performance diff} For any $i\in\playerSet$, $\mu\in\mathcal{P}(S)$, 
 ${\pi}'_i, {\pi}_i \in \Pi_i$, and $\pi_{-i}\in \Pi_{-i}$,
$$
\begin{aligned}
& V_i(\mu, {{\pi}'_i, \pi_{-i}})-V_i(\mu, {{\pi}_i, \pi_{-i}})  =\frac{1}{1-\discount} \sum_{s, a_i} d_\mu^{{\pi}'_i, \pi_{-i}}(s) \cdot\left({\pi}'_i(a_i|s)-{\pi}_i(a_i|s)\right){Q}^{{\pi}_i, \pi_{-i}}_i\left(s, a_i \right). 
\end{aligned}
$$
\end{lemma}

\begin{lemma}[Policy improvement] \label{lemma: policy improve} For Markov $\alpha$-potential game \eqref{eqn: def eps MPG} with any state distribution $\nu{\in \mathcal{P}(S)}$, the $\alpha$-potential function $\Phi(\nu, \pi)$ at two consecutive policies $\pi^{(t+1)}$ and $\pi^{(t)}$ in Algorithm \ref{alg: PolicyGradient} satisfies
\begin{equation*}
    \begin{aligned}
     \text{(i)}\,& \Phi(\nu, {\pi^{(t+1)}})-\Phi(\nu, {\pi^{(t)}}) + N^2 \alpha \geq -\frac{4 \eta^2 \bar{A}^2N^2}{(1-\discount)^5}   +\frac{1}{2 \eta(1-\discount)} \sum_{i\in \playerSet ,s\in S} d_\nu^{\pi_i^{(t+1)}, \pi_{-i}^{(t)}}(s)\left\|\pi_i^{(t+1)}(s) -\pi_i^{(t)}(s)\right\|^2;\\
      \text{(ii)}\, &\Phi(\nu, {\pi^{(t+1)}})-\Phi(\nu, {\pi^{(t)}}) + N^2 \alpha \geq 
       \frac{1}{2 \eta(1-\discount)} \left(1-\frac{4 \eta \kappa_\nu^3 \bar{A} N}{(1-\discount)^4}\right) \sum_{i\in \playerSet, s\in S} d_\nu^{\pi_i^{(t+1)}, \pi_{-i}^{(t)}(s)} 
     \cdot \left\|\pi_i^{(t+1)}(s)-\pi_i^{(t)}(s)\right\|^2.
 \end{aligned}
 \end{equation*}

\end{lemma}


\subsubsection*{Proof of Theorem \ref{thm: PG regret bound}}
Using the variational characterization of projection operation in \eqref{eqn: policy grad Q update}, we note that  
for any $\pi_i^{\prime} \in \Pi_i$,
$$
\left\langle\pi_i^{\prime}(s)-\pi_i^{(t+1)}(s), \eta {\bQ}_i^{(t)}(s)-\pi_i^{(t+1)}(s)+\pi_i^{(t)}(s)\right\rangle_{A_i} \leq 0. 
$$
Therefore, for any $\pi_i^{\prime} \in \Pi_i$,
$$
\begin{aligned}
 \left\langle\pi_i^{\prime}(s)-\pi_i^{(t)}( s), {\bQ}_i^{(t)}(s)\right\rangle_{A_i} 
&= \left\langle\pi_i^{\prime}(s)-\pi_i^{(t+1)}(s), {\bQ}_i^{(t)}(s)\right\rangle_{A_i} +\left\langle\pi_i^{(t+1)}(s)-\pi_i^{(t)}(s), {\bQ}_i^{(t)}(s)\right\rangle_{A_i} \\
&\leq  \frac{1}{\eta}\left\langle\pi_i^{\prime}(s)-\pi_i^{(t+1)}(s), \pi_i^{(t+1)}(s)-\pi_i^{(t)}(s)\right\rangle_{A_i}\\
&\quad+\left\langle\pi_i^{(t+1)}(s)-\pi_i^{(t)}(s), {\bQ}_i^{(t)}(s)\right\rangle_{A_i}.
\end{aligned}
$$
Note that for any two probability distributions $p_1$ and $p_2$, $\|p_1 - p_2 \|\leq \| p_1 - p_2 \|_1 \leq 2.$ Therefore,
    \begin{align}
    \left\langle\pi_i^{\prime}(s)-\pi_i^{(t)}(s), {\bQ}_i^{(t)}(s)\right\rangle_{A_i}  
& \leq  \frac{2}{\eta}\left\|\pi_i^{(t+1)}(s)-\pi_i^{(t)}(s)\right\| +\left\|\pi_i^{(t+1)}(s)-\pi_i^{(t)}(s)\right\|\left\|{\bQ}_i^{(t)}(s)\right\| \notag\\
&\leq \frac{3}{\eta}\left\|\pi_i^{(t+1)}( s)-\pi_i^{(t)}( s)\right\|\label{eqn: calculation1 nash}
,
\end{align}
where the last inequality is due to $\left\|{\bQ}_i^{(t)}(s)\right\| \leq \frac{\sqrt{\bar{A}}}{1-\discount} \text { and } \eta \leq \frac{1-\discount}{\sqrt{\bar{A}}}$. Hence, by Lemma \ref{lemma: performance diff} and \eqref{eqn: calculation1 nash},
    \begin{align*}
       T&\cdot \text{Nash-regret}(T) =  \sum_{t=1}^T \max _{i\in\playerSet, \pi_i^{\prime}} V_i(\mu,{\pi_i^{\prime}, \pi_{-i}^{(t)}})-V_i(\mu,{\pi^{(t)}})\\
      {=}& \sum_{t=1}^T \max _{\pi_i^{\prime}} \sum_{s, a_i} \frac{d_\mu^{\pi_i^{\prime}, \pi_{-i}^{(t)}}(s)}{1-\discount} (\pi_i^{\prime}\left(a_i|s\right)-\pi_i^{(t)}\left(a_i|s\right))\bQ_i^{(t)}\left(s, a_i\right)\\
    {\leq} &\frac{3}{\eta(1-\discount)} \sum_{t=1}^T \sum_s d_\mu^{\pi_i^{\prime}, \pi_{-i}^{(t)}}(s)\left\|\pi_i^{(t+1)}(s)-\pi_i^{(t)}(s)\right\|, 
    \end{align*}
    where in the second line we slightly abuse the notation $i$ to represent $\arg\max_i$ and in the last line we slightly abuse the notation $\pi_i'$ to represent $\arg\max_{\pi_i'}.$ Now, continuing the above calculation with an arbitrary $\nu \in \mathcal{P}( {S})$ and using 
    $$\frac{d_\mu^{\pi_i^{\prime}, \pi_{-i}^{(t)}}(s)}{d_\nu^{\pi_i^{(t+1)}, \pi_{-i}^{(t)}}(s)} \leq \frac{d_\mu^{\pi_i^{\prime}, \pi_{-i}^{(t)}}(s)}{(1-\discount) \nu(s)} \leq \frac{\sup _{\pi \in \Pi}\left\|d_\mu^{\pi } / \nu\right\|_{\infty}}{1-\discount}$$
    to get:
    \begin{align}\label{eqn: bound T nashT}
       &T\cdot \text{Nash-regret}(T)\notag\\
      & {\leq} \frac{ 3 \sqrt{\sup _{\pi \in \Pi}\left\|d_\mu^\pi / \nu\right\|_{\infty}}}{\eta(1-\discount)^{\frac{3}{2}}} \sum_{t=1}^T \sum_s \sqrt{d_\mu^{\pi_i^{\prime}, \pi_{-i}^{(t)}}(s) d_\nu^{\pi_i^{(t+1)}, \pi_{-i}^{(t)}}(s)} \cdot\left\|\pi_i^{(t+1)}(s)-\pi_i^{(t)}(s)\right\|\notag\\
       &{\leq} \frac{3 \sqrt{\sup _{\pi \in \Pi}\left\|d_\mu^\pi / \nu\right\|_{\infty}}}{\eta(1-\discount)^{\frac{3}{2}}} \sqrt{\sum_{t=1}^T \sum_s d_\mu^{\pi_i^{\prime}, \pi_{-i}^{(t)}}(s)}\cdot \sqrt{\sum_{t=1}^T \sum_{i=1}^{N} \sum_s d_\nu^{\pi_i^{(t+1)}, \pi_{-i}^{(t)}(s)}\left\|\pi_i^{(t+1)}(s)-\pi_i^{(t)}(s)\right\|^2},
    \end{align}
    where the last inequality follows from the Cauchy-Schwarz inequality and replacing $i$ ($\arg\max_i$) by the sum over all players.
    There are two choices to proceed beyond \eqref{eqn: bound T nashT}:\\
1) Fix $\epsilon > 0.$ Take $\nu^*_\epsilon \in \mathcal{P}(S)$ such that $\sup_{\pi\in\Pi} \left\| {d_\mu^\pi}/{\nu_\epsilon^*}\right\|_\infty - \epsilon \leq \inf_{\nu \in \mathcal{P}(S)} \sup_{\pi \in \Pi}\left\| {d_\mu^\pi}/{\nu_\epsilon^*}\right\|_\infty.$ Then apply Lemma \ref{lemma: policy improve} (i) and the fact $| \Phi (\nu, \pi) -  \Phi (\nu, {\pi'}) | \leq C_\Phi$ for any $\pi, \pi' \in \Pi, \nu \in \mathcal{P}(S)$
to get
\begin{align*}
   \text{Nash-regret}(T) 
   \leq &\frac{3}{T} \Bigg(\frac{ 2 (\widetilde{\kappa}_\mu+ \epsilon) T (C_{\Phi} + N^2 \alpha \cdot T)}{\eta(1-\discount)^2} 
   + \frac{8(\widetilde{\kappa}_\mu + \epsilon) \eta T^2 \bar{A}^2 N^2}{(1-\discount)^7}\Bigg)^\frac{1}{2}.
\end{align*}
By letting $\epsilon$  to $0$ and taking step size $\eta =  \frac{(1-\discount)^{2.5} \sqrt{C_\Phi + N^2 \alpha T}}{2 N \bar{A} \sqrt{T}}$, we have
$$
\text{Nash-regret}(T) \leq \frac{3 \cdot 2^{\frac{3}{2}} \sqrt{{\tilde{\kappa}_\mu} \bar{A}N} }{(1-\discount)^{\frac{9}{4}} } \left( \frac{C_\Phi }{T}+ N^2 \alpha \right)^\frac{1}{4}. 
$$
2) We can also proceed \eqref{eqn: bound T nashT}  with Lemma \ref{lemma: policy improve} (ii) and $\eta \leq \frac{(1-\discount)^4}{8 \kappa_\nu^3 N \bar{A}}$ to get 
\begin{align*}
     \text{Nash-regret}(T)
     \leq 6 \sqrt{\frac{\sup _{\pi \in \Pi}\left\|\frac{d_\mu^\pi}{\nu} \right\|_{\infty}  (C_{\Phi} + N^2 \alpha \cdot T)}{\eta T (1-\discount)^2}}.
\end{align*}
We next discuss two special choices of $\nu$ for proving our bound. First, if $\nu=\mu$, then $\eta \leq \frac{(1-\discount)^4}{8 \kappa_\mu^3 N \bar{A}}$. By letting $\eta=\frac{(1-\discount)^4}{8 \kappa_\mu^3 N \bar{A}}$, the last square root term can be bounded by 
$
\mathcal{O}\Bigg(\sqrt{\frac{\kappa_\mu^4 N \bar{A}  (C_{\Phi}+ N^2 \alpha \cdot T)}{T(1-\discount)^6}}\Bigg).
$
Second, if $\nu=\frac{1}{|S|} \mathbf{1}$, the uniform distribution over ${S}$, then $\kappa_\nu \leq \frac{1}{S}$, which allows a valid choice $\eta=\frac{(1-\discount)^4}{8 |S|^3 N \bar{A}} \leq \frac{(1-\discount)^4}{8 \kappa_\nu^3 N \bar{A}}$. Hence, we can bound the last square root term by 
$
\mathcal{O}\left(\sqrt{\frac{|S|^4 N \bar{A}  (C_{\Phi}+ N^2 \alpha \cdot T)}{T (1-\discount)^6}}\right).
$
Since $\nu$ is arbitrary, combining these two special choices completes the proof.

\subsection{Sequential maximum improvement algorithm}\label{sssec: AlgBR}
Let us first fix some  notations. Associated with any Markov game \(\game\), we define \emph{smoothed} (or regularized) Markov game $\Gtilde$, where the expected one-stage payoff of each player $i$ with state $s$  under the joint policy $\pi$ is \(
    \utilde_i(s, \pi)= \mathbb{E}_{a\sim \pi(s)}[u_i(s, a)] - \tau \sum_{j\in \playerSet}\nu_j(s,\pi_j),\) where 
\(
 \nu_j(s, \pi_j) \coloneqq \sum_{a_j \in A_j} \pi_j(a_j|s) \log(\pi_j(a_j|s))
\) 
is the entropy function,
and $\tau > 0$ denotes the regularization parameter. With the smoothed one-stage payoffs, the expected total discounted infinite horizon payoff of player $i$ under policy \(\pi\) is given by
\begin{equation}\label{eqn: smoothed value function}
    \tilde{V}_i(s, \policy) = \avg_{\pi}\Big[\sum_{k=0}^{\infty}\discount^k \big(\stagePayoff_i(s^k,a^k) - \tau\sum_{j\in \playerSet }\nu_j(s^k,\pi_j)\big)|s^0=s\Big],
\end{equation}
for every \(s\in S\). The \emph{smoothed} (or entropy-regularized) \(Q\)-function is given by 
{   
\begin{equation}\label{eqn: smooth Q}
    \begin{aligned}
  \tilde{Q}^\pi_i(s,a_i) &= \sum_{a_{-i}\in A_{-i}} \pi_{-i}(a_{-i}|s)\Big( {{u}_i(s,a_i,a_{-i}) }- \tau \sum_{j\in I_N }\nu_j(s,\pi_j) 
 +\discount\sum_{s'\in S}P(s'|s,a)\tilde{V}_i(s',\pi)\Big). 
\end{aligned}
\end{equation}}

Algorithm \ref{alg: SequentialPureBRMaxPlayer} 
has two main components: first, it computes the optimal one-stage policy update using the smoothed $Q$-function. Here, the vector of smoothed $Q$-functions is denoted by \(\tilde{\bQ}^\pi_i(s) = (\tilde{Q}_i^\pi(s,a_i))_{a_i\in A_i}\). Second, it selects the player who achieves the maximum improvement in the current state to adopt her one-stage policy update, with the policy for the remaining players and the remaining states unchanged.  
More specifically, 
the algorithm iterates for \(T\) time steps. In every time step \(t\in [T-1]\), based on the current policy profile $\pi^{(t)}$, {    
we abuse the notation to use $ \tilde{Q}_i^{(t)} $ to denote $\tilde{Q}_i^{\pi^{(t)}}$ and $ \tilde{\bQ}_i^{(t)} $ to denote $\tilde{\bQ}_i^{\pi^{(t)}}$. 
The expected smoothed \(Q\)-function of player $i$ is computed as \(\tilde{Q}_i^{(t)}(s,\pi_i) = \sum_{a_i\in A_i} \pi_i(a_i|s)\tilde{Q}^{(t)}_i(s,a_i)\) for all $s \in S$ and all $i \in \playerSet$.} Then, each player computes her one-stage best response strategy by maximizing the smoothed \(Q\)-function: for every \(i\in \playerSet, a_i\in A_i, s\in S,\) 
\begin{align}\label{eq: SmoothedBR}
\text{BR}_{i}^{(t)}(a_i|s) &=  \lr{\underset{\pi_i'\in \Pi_i} {\arg\max} \ \left(\tilde{Q}^{(t)}_i(s,\pi_i')- \tau \nu_i(s,\pi_i')\right)}_{a_i} \notag\\
&=\frac{\exp(\tilde{Q}_i^{(t)}(s,a_i)/\tau)}{\sum_{a_i'\in A_i}\exp(\tilde{Q}_i^{(t)}(s,a_i')/\tau) },
\end{align}
and its maximum improvement of smoothed \(Q\)-function value in comparison to current policy is \begin{align}\label{eq:max_improv}
\Delta_i^{(t)}(s) &= \max_{\pi_i'\in \Pi_i}\lr{\tilde{Q}^{(t)}_i(s,\pi_i')- \tau \nu_i(s,\pi_i')} - \lr{\tilde{Q}^{(t)}_i(s,\pi_i^{(t)})- \tau 
\nu_i(s,\pi_i^{(t)})}, \quad \forall s\in S.
\end{align}
{Note that computing \(\Delta_i^{(t)}\) is straightforward as the maximization in \eqref{eq:max_improv} is attained at \(\text{BR}_{i}^{(t)}(s)\) (cf. \eqref{eq: SmoothedBR}).}

If the maximum improvement $\Delta_i^{(t)}(s) \leq 0$ for all $i \in \playerSet$ and all $s \in S$, then the algorithm terminates and returns the current policy profile $\pi^{(t)}$. Otherwise, the algorithm chooses a tuple of player and state $(\imax, \smax)$ associated with the maximum improvement value $\Delta_i^{(t)}(s)$, and updates the policy of player $\imax$ in state $\smax$ with her one-stage best response strategy\footnotemark{}.\footnotetext{Any tie-breaking rule can be used here if the maximum improvement is achieved by more than one tuple.} The policies of all other players and other states remain unchanged. 
\begin{remark}
    Using entropy regularization in \eqref{eq: SmoothedBR} has several advantages: (i) unlike Algorithm \ref{alg: PolicyGradient}, it avoids projection over simplex which can be costly in large-scale problems; (ii) it ensures that the optimizer is unique. 
\end{remark}

\begin{remark}
Algorithm \ref{alg: SequentialPureBRMaxPlayer} is reminiscent of the ``Nash-CA'' algorithm\footnotemark{}\footnotetext{Unlike this paper, the Nash-CA Algorithm in \cite{song2021can} was proposed in the context of finite horizon Markov potential games.} proposed in \cite{song2021can}, which requires each player to \textit{sequentially} compute the best response policy using an RL algorithm in each iteration, while keeping the strategies of other players fixed. Such sequential best response algorithms are known to ensure finite improvement in the potential function value in potential games \cite{monderer1996potential}, which ensures convergence. Meanwhile, Algorithm \ref{alg: SequentialPureBRMaxPlayer} does not compute the best response strategy in the updates. Instead, it only computes a smoothed \emph{one-step} optimal deviation, as per \eqref{eq: SmoothedBR}, for the current state. The policies for the remaining states and other players are unchanged. The analysis of such one-step deviation-based dynamics is non-trivial and requires new techniques, as highlighted in the next section.

\end{remark}
\begin{remark}
     While Algorithm \ref{alg: PolicyGradient} can be run independently by each player in a decentralized fashion, Algorithm \ref{alg: SequentialPureBRMaxPlayer} is centralized as  players do not update their policies simultaneously. Comparing Nash regret in Theorems \ref{thm: PG regret bound} and \ref{thm: SeqMaxImprovement}, it is evident that the coordination in Algorithm \ref{alg: SequentialPureBRMaxPlayer} ensures better scaling of regret with respect to the number of players.
\end{remark}

\begin{algorithm}\caption{Sequential Maximum  Improvement Algorithm}
\begin{algorithmic}
\STATE{\textbf{Input:} Smoothness parameter \(\tau\), for every \(i\in \playerSet , a_i\in A_i, s\in S,\) set \(\pi_i^{(0)}(a_i|s) = 1/|A_i|\).}

\FOR{\(t = 0,1,2,...,T-1\)}
\STATE{Compute the maximum improvement of smoothed \(Q\)-function $\{\Delta_i^{(t)}(s)\}_{i \in \playerSet , s\in S}$ as in \eqref{eq:max_improv}.}

\IF{ \(\Delta_i^{(t)}(s) \leq 0\) for all $i \in \playerSet $ and all $s \in S$
}
\STATE{return $\pi^{(t)}.$
}
\ELSE
\STATE{Choose the tuple \((\imax, \smax)\) with the maximum improvement}
\begin{align}\label{eq: UpdateSequence}
(\imax, \smax) \in \underset{i\in \playerSet , s \in S }{\arg\max} ~ \Delta_i^{(t)}(s),
\end{align}
and update policy
\begin{align}\label{eq: PolicyMaxMaxImprove}
\pi_{\imax}^{(t+1)}(a|\smax)  =& \text{BR}_{\imax}^{(t)}(a|\smax), ~ \forall a\in A_{\imax}, 
 \\
    \pi_i^{(t+1)}(s) =& \pi_i^{(t)}(s)~ \forall (i, s) \neq (\imax, \smax).\notag
\end{align}
\ENDIF
\ENDFOR
\end{algorithmic}\label{alg: SequentialPureBRMaxPlayer}
\end{algorithm}

\begin{theorem}\label{thm: SeqMaxImprovement}
    Consider a Markov \(\mnpg\)-potential game with an \(\mnpg\)-potential function \(\Phi\) and initial state distribution \(\mu\) such that  \( \bar{\mu}\coloneqq \min_{s\in S}\mu(s) > 0\). 
    Denote \(\bar{A} \coloneqq \max_{i\in \playerSet }|A_i|\) and \(C \coloneqq \max_{i\in \playerSet } \| u_i\|_{\infty}\). 
    Then the policy updates generated from Algorithm \ref{alg: SequentialPureBRMaxPlayer} with parameter
    \begin{align} \label{eq: TauDef}
        \tau = \frac{1}{N}\Bigg(&\log(\bar{A}) + \frac{\log(\bar{A})}{\sqrt{\mnpg + \frac{C_{\Phi}}{T}}}\sqrt{\frac{2\log(\bar{A})}{(1-\discount)}}\sqrt{\frac{N}{T}} + \frac{2\sqrt{\bar{\mu}}(1-\discount)\log(\bar{A})}{{8C\sqrt{\bar{A}}}\sqrt{\mnpg + \frac{C_{\Phi}}{T}}} \Bigg)^{-1}
    \end{align}
    has the $\text{Nash-regret}(T)$ bounded by
    \begin{align*}
    \mathcal{O}\Bigg(\frac{\sqrt{N^{3/2}\bar{A}}\log(\bar{A})}{(1-\discount)^{5/2}\sqrt{\bar{\mu}}}
        \max\bigg\{\Big(\mnpg + \frac{C_{\Phi}}{T}\Big)^{\frac{1}{2}} , \Big(\mnpg + \frac{C_{\Phi}}{T}\Big)^{\frac{1}{4}} \bigg\}\Bigg),
    \end{align*}
where $C_\Phi >0$ is a constant satisfying $\left|\Phi(\mu,\pi)-\Phi(\mu,\pi^{\prime})\right| \leq C_{\Phi}$ for any $\pi, \pi^{\prime}\in \Pi, \mu \in \mathcal{P}(S).$ 
\end{theorem}

In the special case of $\alpha=0$,
       Theorem \ref{thm: SeqMaxImprovement} provides a Nash-regret bound of Algorithm \ref{alg: SequentialPureBRMaxPlayer} for the case of MPGs.
       
 To prove Theorem \ref{thm: SeqMaxImprovement}, we first develop a smoothed version of the multi-agent performance difference lemma (Lemma \ref{lem: performanceDiffPG}). This lemma bounds the difference in the smoothed value function $\Tilde{V}_i$ by the changes in policy $\pi_i$, which is further bounded by the maximum improvements $\Delta_i^{(t)}$. 
Lemma \ref{lem: PerturbGameInitialGameDiff} bounds the discrepancy between the value function $V_i$ and the smoothed value function $\Tilde{V}_i$. Lemma \ref{lem: performanceDiffPG} and \ref{lem: PerturbGameInitialGameDiff} together implies that the Nash-regret of Algorithm \ref{alg: SequentialPureBRMaxPlayer} is bounded by $\Delta_i^{(t)}$ \eqref{eq:max_improv}. Finally, Lemma \ref{lemma:technical} establishes $\Delta_i^{(t)}$ can be bounded by policy updates, which in turn, are bounded by $\alpha$ and the difference in the $\alpha$-potential function $\Phi$.

\begin{lemma}[Smoothed performance difference]\label{lem: performanceDiffPG}
For any $i\in \playerSet,$ $\mu\in \mathcal{P}(S)$, \(\pi_i,\pi_i'\in\Pi_i, \pi_{-i}\in \Pi_{-i}\),  
\begin{align*}
\tilde{\vFunc}_i(\mu,\policy)-&\tilde{\vFunc}_i(\mu,\policy') = \frac{1}{1-\discount}\sum_{s'\in S}d^{\pi}_{\mu}(s') \Big((\pi_i(s') - \pi_i'(s'))^\top 
\cdot\tilde{\bQ}^{\pi'}_i(s')+\tau\ti(s',\policy_i')-\tau\ti(s',\policy_i)\Big),
\end{align*}
where \(\pi=(\pi_i,\pi_{-i})\), and \(\pi'=(\pi_i',\pi_{-i})\).
\end{lemma}

\begin{lemma}\label{lem: PerturbGameInitialGameDiff}
For any \(i\in \playerSet , \mu\in \mathcal{P}(S), \pi_i, \pi_i'\in \Pi_i, \pi_{-i}\in\Pi_{-i}\), 
\(
    \Big|\vFunc_i(\mu, \policy_i,\pi_{-i})-\vFunc_i(\mu, \pi_i',\pi_{-i}) -  (\tilde{\vFunc}_i(\mu, \policy_i,\pi_{-i}) -\tilde{\vFunc}_i(\mu, \pi_i',\pi_{-i}))\Big|\leq  \frac{2\tau N \log(\bar{A})}{1-\discount}.
\)
\end{lemma}

\begin{lemma}The following inequalities hold:
\begin{enumerate}\label{lemma:technical}
\item[(1)]  \(\Delta_{\imax}^{(t)}(\smax) \leq \frac{4C\sqrt{\bar{A}}(1+\tau N\log(\bar{A}))}{1-\discount}\|\pi_{\imax}^{(t+1)}(\smax) - \pi_{\imax}^{(t)}(\smax) \|_2,\) for any \(t \in [T]\).
\item[(2)] 
$\sum_{t=0}^{T-1}\|\pi_{\imax}^{(t+1)}(\smax)-\pi_{\imax}^{(t)}(\smax)\|_2^2\leq \frac{2}{\tau\bar{\mu}}\Big(|\Phi(\mu,\pi^{(T)}) - \Phi(\mu,\pi^{(0)})| + \mnpg T  + \frac{2\tau N\log(\bar{A})}{1-\discount}\Big).$
\end{enumerate}
\end{lemma}

\subsubsection*{Proof of Theorem \ref{thm: SeqMaxImprovement}}
First, we bound the instantaneous regret \(R_i^{(t)}\) for any arbitrary player \(i\in \playerSet \) at time \(t\in [T]\). Recall that  
\(
\begin{aligned}
     {R}_i^{(t)} &= {V}_i(\mu,\pi_i^\dagger, \pi_{-i}^{(t)})-{V}_i(\mu,\pi^{(t)}), 
\end{aligned}
\)
where \(\pi_i^\dagger \in 
\arg\max _{\pi_i^{\prime}\in \Pi_i} {V}_i(\mu,\pi_i^{\prime}, \pi_{-i}^{(t)})\).
By Lemma \ref{lem: PerturbGameInitialGameDiff},
\[
   {R}_i^{(t)} \leq   \tilde{V}_i(\mu,\pi_i^\dagger, \pi_{-i}^{(t)})-\tilde{V}_i(\mu,\pi^{(t)}) + \frac{2\tau N\log(\bar{A})}{(1-\discount)}.
\]
Next, note that for any \(i\in \playerSet , \mu\in \mathcal{P}(S)\), by Lemma \ref{lem: performanceDiffPG},  
\begin{align*}
&\tilde{V}_i(\mu,{\pi_i^\dagger,\pi_{-i}^{(t)}})-  \tilde{V}_i(\mu,\pi_i^{(t)},\pi_{-i}^{(t)})  \\ 
&{\leq} \frac{1}{1-\discount}  \sum_{s\in S} d_{\mu}^{\pi_i^\dagger,\pi_{-i}^{(t)}}(s)
\bigg(
\tau (\nu_i(s,\pi_i^{(t)}) -\nu_i(s,\pi_i'))
+
\max_{\pi_i'}\sum_{a_i\in A_i}\left(\left(  \pi_i'(a_i|s) - \pi_i^{(t)}(a_i|s) \right) \tilde{Q}_i^{(t)}(s,a_i) \right) 
\bigg) 
\\
&\stackrel{(a)}{=} \frac{1}{1-\discount}  \sum_{s\in S} d_{\mu}^{\pi_i^\dagger,\pi_{-i}^{(t)}}(s)\Delta_i^{(t)}(s) 
\\&
\stackrel{(b)}{\leq} \frac{1}{1-\discount}  \sum_{s\in S} d_{\mu}^{\pi_i^\dagger,\pi_{-i}^{(t)}}(s)\Delta_{\imax}^{(t)}(\smax) {=} \frac{1}{1-\discount}  \left(\Delta_{\imax}^{(t)}(\smax)\right), 
\end{align*}
where
\((a)\) is by \eqref{eq:max_improv}, \((b)\) holds since \(\Delta_i^{(t)}(s)\leq \Delta_{\imax}^{(t)}(\smax)\) for all \(i\in \playerSet , s\in S\). To summarize,
\[
    {R}_i^{(t)}\leq \frac{1}{1-\discount} \lr{\Delta_{\imax}^{(t)}(\smax)+2\tau N\log(\bar{A})}.
\]
    Then  by Lemma \ref{lemma:technical} (1),
\begin{align}
    &\text{Nash-regret}(T)\leq \frac{1}{T(1-\discount)} \sum_{t\in [T]}\lr{\Delta_{\imax}^{(t)}(\smax)+2\tau N\log(\bar{A})} \notag\\ 
    &{\leq} \frac{2\tau N\log(\bar{A})}{(1-\discount)} + \frac{4C\sqrt{\bar{A}}(1+\tau N\log(\bar{A}))}{T(1-\discount)^2}
   \cdot\sum_{t\in [T]} \Big\|\pi_{\imax}^{(t+1)}(\smax) - \pi_{\imax}^{(t)}(\smax) \Big\|_2 \notag \\ 
    &{\leq} \frac{2\tau N\log(\bar{A})}{(1-\discount)} + \frac{4C\sqrt{\bar{A}}(1+\tau N\log(\bar{A}))}{\sqrt{T}(1-\discount)^2}
  \cdot\Big(\sum_{t\in [T]}\Big\| \pi_{\imax}^{(t+1)}(\smax) - \pi_{\imax}^{(t)}(\smax)\Big\|_2^2\Big)^\frac{1}{2},\label{eqn: bound sq part1}
\end{align}
where the last inequality follows from Cauchy-Schwarz inequality. 
For ease of exposition, define $
    D_1 \coloneqq \frac{8C\sqrt{\bar{A}}}{\sqrt{\bar{\mu}}(1-\discount)^2},  D_2 \coloneqq  \sqrt{\mnpg + \frac{C_{\Phi}}{T}}$, and  $D_3 \coloneqq  \sqrt{\frac{2\log(\bar{A})}{(1-\discount)}}.$
Then by  Lemma \ref{lemma:technical} (2),
\begin{align*}
     \eqref{eqn: bound sq part1}&{\leq} \frac{D_1(1+\tau N\log(\bar{A}))}{\sqrt{\tau}}\sqrt{D_2^2 + \frac{\tau N}{T} D_3^2} + \tau N D_3^2
    \\
    &{\leq} \frac{D_1(1+\tau N\log(\bar{A}))}{\sqrt{\tau}}\lr{D_2 + \sqrt{\frac{\tau N}{T} }D_3 }  + {\tau N}{D_3^2},
\end{align*}
where the last inequality follows from the fact that for any two positive scalars \(x,y\), \(\sqrt{x+y}\leq \sqrt{x}+\sqrt{y}\). 
Setting \(\tau\) as per \eqref{eq: TauDef} ensures that \(\tau < \sqrt{\tau}\) as \(\tau \leq 1\).
Thus, 
\begin{align*}
    &\text{Nash-regret}(T) \leq \frac{D_1D_2}{\sqrt{\tau}} + \frac{D_1D_3\sqrt{N}}{\sqrt{T}} +\sqrt{\tau}N\lr{ D_1D_2\log(\bar{A})  + D_1D_3\log(\bar{A})\sqrt{\frac{N}{T}} + D_3^2}.
\end{align*}
Plugging in the value of \(\tau\) as per \eqref{eq: TauDef} we obtain,
\begin{align*}
    \text{Nash-regret}(T) &\leq \sqrt{N}\Bigg(D_1^2D_2^2\log(\bar{A}) + D_1^2D_2D_3\log(\bar{A})\sqrt{\frac{N}{T}} + D_1D_2D_3^2\Bigg)^{\frac{1}{2}} + \frac{D_1D_3\sqrt{N}}{\sqrt{T}} \\ 
    &\leq D_1D_2\sqrt{N}\sqrt{\log(\bar{A})} + D_1\sqrt{D_2D_3\log(\bar{A})} \frac{N^\frac{3}{4}}{T^{\frac{1}{4}}} + \sqrt{D_1D_2}D_3\sqrt{{N}}+ \frac{D_1D_3\sqrt{N}}{\sqrt{T}}. 
    \end{align*}
    Note that \(D_3\geq 1\) and additionally, we assume that \(D_1 \geq 1\) (choose large enough  \(C\) that ensures this). Then,
    \begin{align*}
    \text{Nash-regret}(T) & \leq D_1D_2D_3\sqrt{N}\sqrt{\log(\bar{A})} + D_1D_3\sqrt{D_2\log(|\bar{A}|)} \frac{N^\frac{3}{4}}{T^{\frac{1}{4}}} + \sqrt{D_2}D_1D_3\sqrt{{N}} + \frac{D_1D_3\sqrt{N}}{\sqrt{T}} 
  \\
    &\leq D_1D_3\sqrt{N\log(\bar{A})} \lr{D_2 + \sqrt{D_2}\lr{1+\lr{\frac{N}{T}}^{\frac{1}{4}}} + \sqrt{\frac{1}{T}}} \\ 
    &\leq
    D_1D_3\sqrt{\log(\bar{A})} N^{\frac{3}{4}}\mathcal{O}(\max\{ D_2, \sqrt{D_2}\}).
    \end{align*}
    The proof is finished by plugging in $D_1,D_2$ and $D_3$.


\section{Numerical experiments}\label{sec: numerics}
This section studies the empirical performance of Algorithms \ref{alg: PolicyGradient} and \ref{alg: SequentialPureBRMaxPlayer}  for Markov congestion game (MCG) and perturbed Markov team game (PMTG) discussed in Section \ref{sec: NearMPG}.
{{}
Although Section \ref{sec: Algorithms} focuses on model-based algorithms, in our numerical study both Algorithm \ref{alg: PolicyGradient} and Algorithm \ref{alg: SequentialPureBRMaxPlayer} are implemented in a model-free manner, where the Q-functions are estimated from samples \cite{ding2022independent, leonardos2021global}.
}
Below are the details of the setup of the experiments.

\textbf{MCG:}
Consider MCG with $N  = 8$ players, where there are $|E| = 4$ facilities $A, B, C, D$ that each player can select from, i.e., $|A_i| = 4$. For each facility $j$, there is an associated state $s_j$: \textit{normal} ($s_j =0$) or \textit{congested} ($s_j = 1$) state, {and the state of the game is $s = (s_j)_{j \in E}$.}
The reward for each
player being at facility $k$ is equal to $w_k^{\text{safe}}$ times the number of players at $k=A, B, C, D.$
We set $w_A^{\text {safe }}=1<w_B^{\text {safe }}=2<w_C^{\text {safe }}=4<w_D^{\text {safe }}=6$, i.e., facility $D$ is most preferable by all players. However, if more than  $N /2$ players find themselves in the same facility,
then this facility transits to the \textit{congested} state, where the reward for each player is reduced by a large constant $c=-100$. To return to the \textit{normal} state, the facility should contain no more than $N /4$ players. 

\textbf{PMTG:} {{}
Consider a game where each player votes for approving or disapproving a project, which is only conducted if a majority of players vote for approval. The state of excitement about the project changes between different rounds depending on the number of players approving it. Mathematically, consider a game with \( N = 16 \) players, where there are two actions per player: \textit{approve} (\( a_i = 1 \)) or \textit{disapprove} (\( a_i = 0 \)). There can be two states of the project: \textit{high} (\( s = 1 \)) and \textit{low} (\( s = 0 \)) levels of excitement for the project. 

The individual reward of player \( i \) is given by
\(
u_i(s,a) = \mathbf{1}_{\sum_{i} a_i \geq N/2} + w_i \mathbf{1}_{\{a_i = s\}} - w_i' a_i,
\)
where the first term represents the common utility derived by everyone if the project is approved, the second term represents the utility derived by a player in approving a high-priority project or disapproving a low-priority project, and the third term corresponds to the cost of approving the project. Here, we set
\(
w_i = 10\kappa \cdot \frac{N + 1 - i}{N}\) and \(w_i' = \kappa \cdot \frac{i + 1}{N}. 
\) Here, parameter \(\kappa\) captures the magnitude of perturbation. 

The state transitions from the \textit{high excitement state} to itself with probability \(\lambda_1\) if more than \(N/4\) players approve it; otherwise, it transitions to itself with probability \( \lambda_2 \). In contrast, the state transitions from the \textit{low excitement state} to high with probability \( \lambda_3 \) if there are at least \( N/2 \) approvers; if there are \( N/2 \) or fewer approvers, it transitions to high with probability \( \lambda_4 \).

}


For both games, we perform episodic updates with \(20\) steps and a discount factor $\discount = 0.99$. We estimate the \(Q\)-functions and the utility functions using the average of mini-batches of
size $10$. 
{{}
For MCG, Figures \ref{fig:a} and \ref{fig:b} illustrate the average number of players taking particular action in different states at the converged values of policy. For example, in the state $(0,0,0,1)$ (denoted by the yellow label in Figure \ref{fig:a} and \ref{fig:b}), facility $D$ is congested, while the other facilities remain in a normal state. In this scenario, only $N/4 = 2$ players select facility $D$ to restore it to a normal state. Simultaneously, $N/2$ players choose facility $C$, which provides the second-highest reward after $D$. The number of players at $C$ is within 
 the congestion threshold ($N/2$), thus ensuring that it remains in a normal state.
}

{{}
For PMTG, we set \(\lambda_1=\lambda_3=1\), \(\lambda_2=\lambda_4=0\) and \(\kappa = 0.1\). Figures \ref{fig:a_pert} and \ref{fig:b_pert} illustrate the average number of players taking particular action in different states at the converged values of policy. For example, in the `high' state of excitement about project (denoted by the red label in Figure \ref{fig:a_pert} and \ref{fig:b_pert}), almost all players will select to approve as it will always remain in high state thereon. Meanwhile, if the state of excitement is `low', then at least half of the players select to approve it so that it transitions to `high' state in future. 
}

{{} Figures \ref{fig:c} and \ref{fig:c_pert} depict the \( L_1 \)-accuracy in the policy space at each iteration, defined as the average distance between the current policy and the final policy of all \( 8 \) players, i.e.,
\(
L_1\text{-accuracy} = \frac{1}{N} \sum_{i \in I} \| \pi_i - \pi_i^{(T)} \|_1.
\)
Figures \ref{fig:c} and \ref{fig:c_pert} show that Algorithm \ref{alg: PolicyGradient} converges faster for PMTG, while Algorithm \ref{alg: SequentialPureBRMaxPlayer} converges faster for MCG.
}
\begin{figure*}[!ht]
        \centering
    \begin{minipage}[l]{0.8\textwidth}
        \centering
        \begin{subfigure}{0.28\textwidth}
    \centering
    \includegraphics[width=\textwidth]{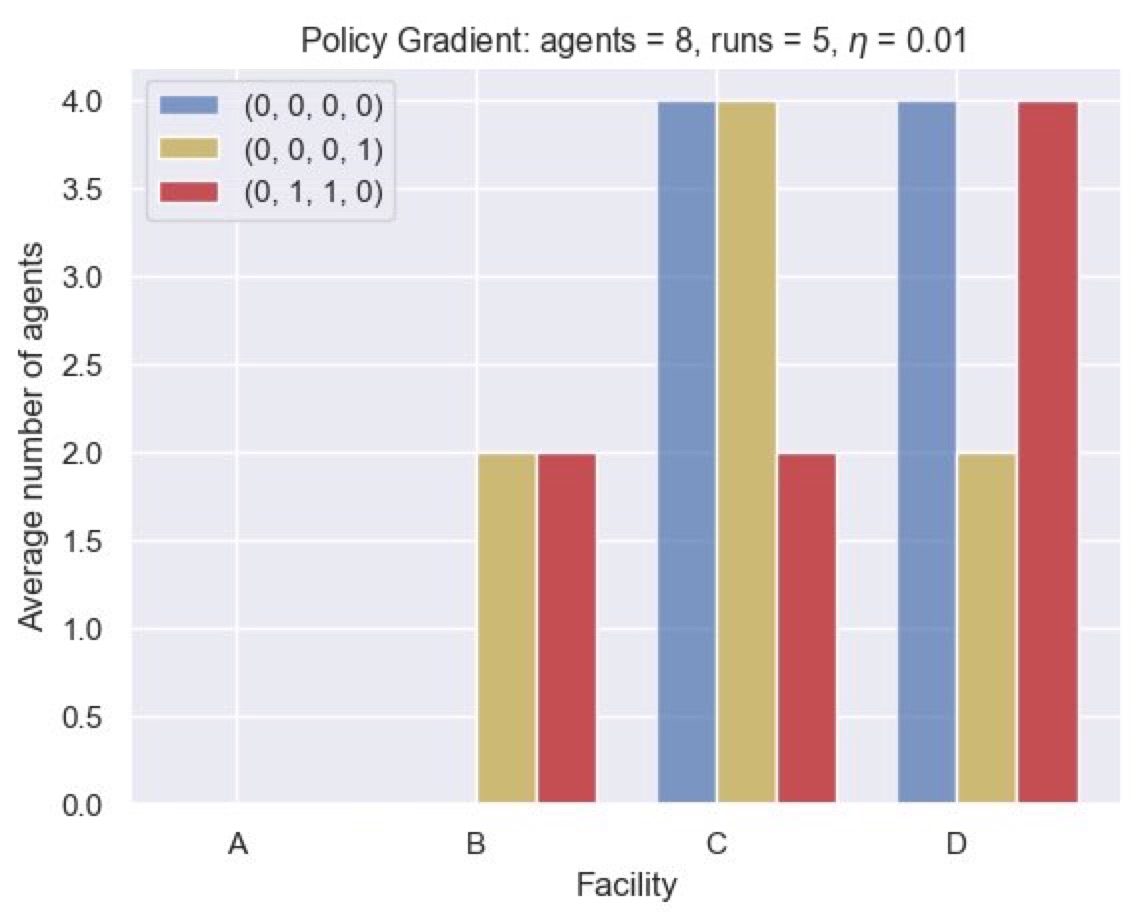}
    \caption{}
    \label{fig:a}
  \end{subfigure}
  \begin{subfigure}{0.28\textwidth}
    \centering
    \includegraphics[width=\textwidth]{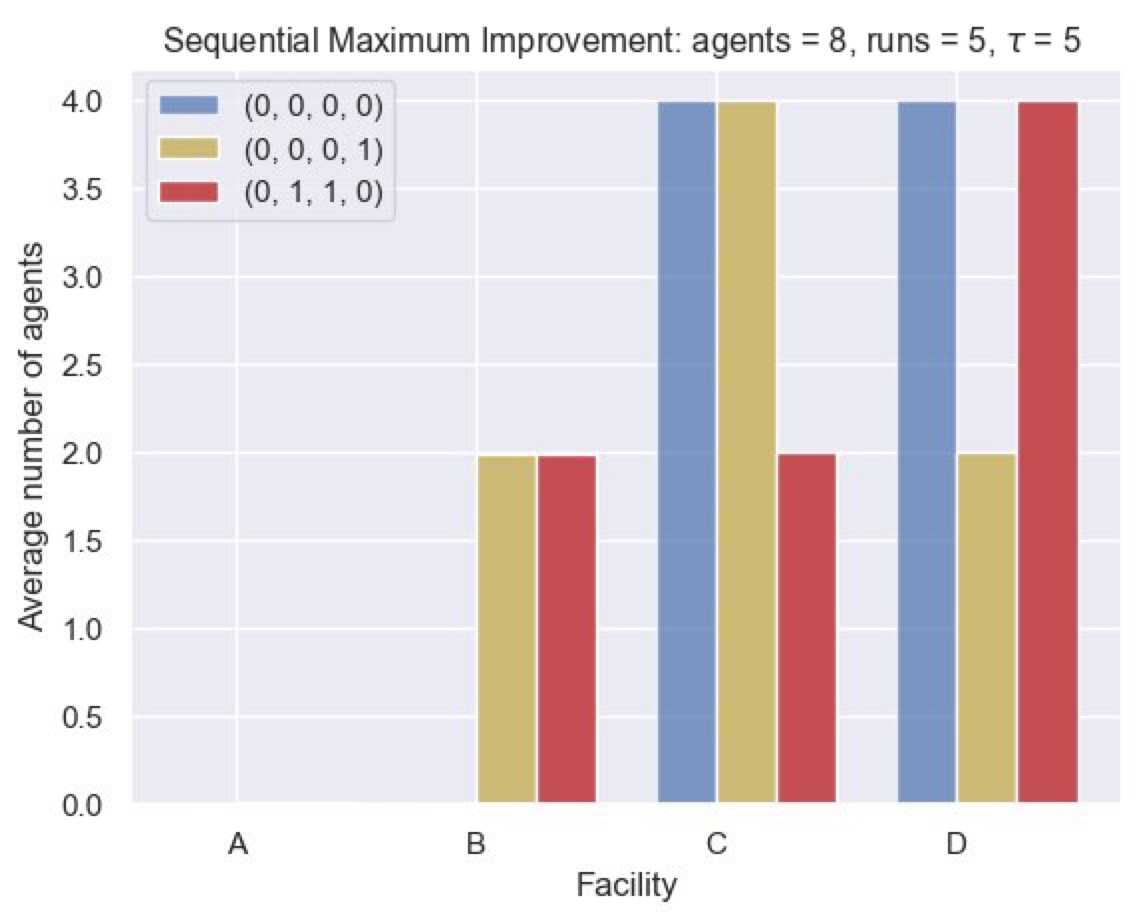}
    \caption{}
    \label{fig:b}
  \end{subfigure}
  \begin{subfigure}{0.3\textwidth}
    \centering
    \includegraphics[width=\textwidth]{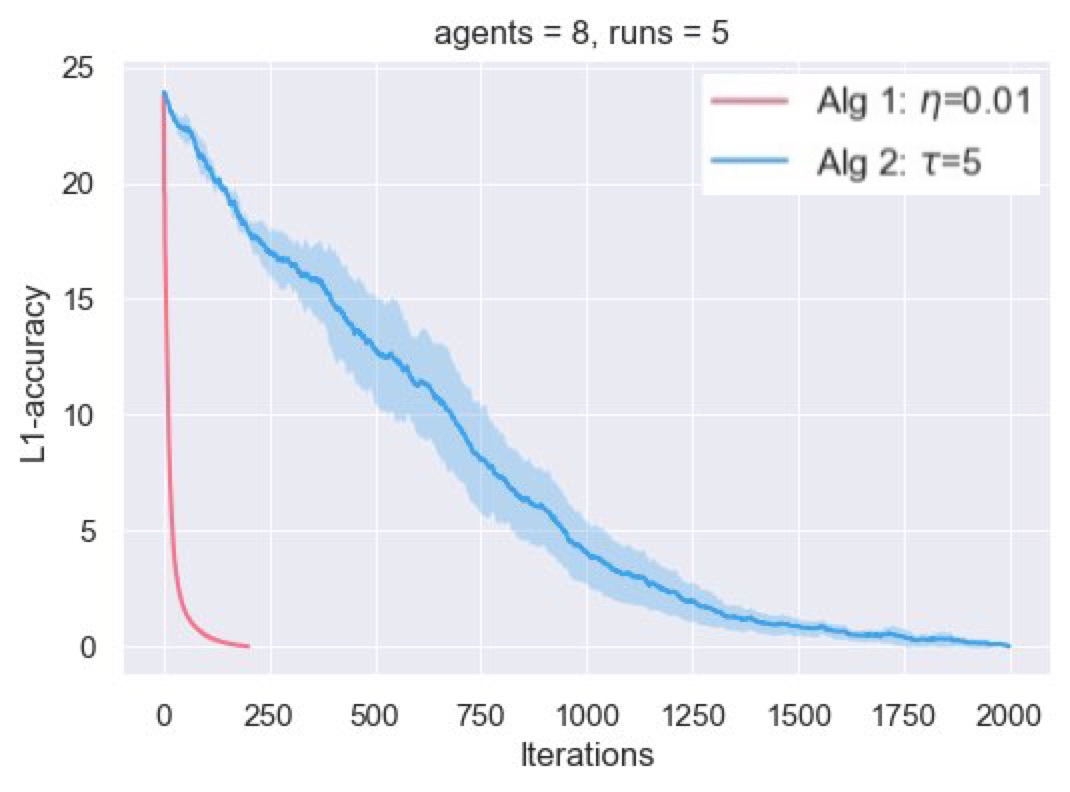}
    \caption{}
    \label{fig:c}
  \end{subfigure}
    \end{minipage}
      \caption{\textbf{Markov congestion game:} (a) and (b) are distributions of players taking four actions in representative states using $\pi^{(T)}$ given by (a) Algorithm \ref{alg: PolicyGradient} with step-size $\eta = 0.01$; (b) Algorithm \ref{alg: SequentialPureBRMaxPlayer} with regularizer $\tau_t = 0.999^t\cdot 5$. 
  (c) is mean L1-accuracy with shaded region of one standard deviation over all runs.} 
      \label{fig:cong}
\end{figure*}
{{}\begin{remark}
    We note that the regret bound proposed in our analysis can be loose. In Figure \ref{fig:regret_bound_loose}, we compare growth of regret bound obtained in our theoretical results with that obtained in experiments, where we observe significant gap between the two quantities. This suggests an interesting direction of future research to develop tighter regret bounds.
\end{remark}
}

\begin{figure*}[!ht]
\centering
    \begin{minipage}[l]{0.8\textwidth}
  \centering
  \begin{subfigure}{0.28\textwidth}
    \centering
    \includegraphics[width=\textwidth]{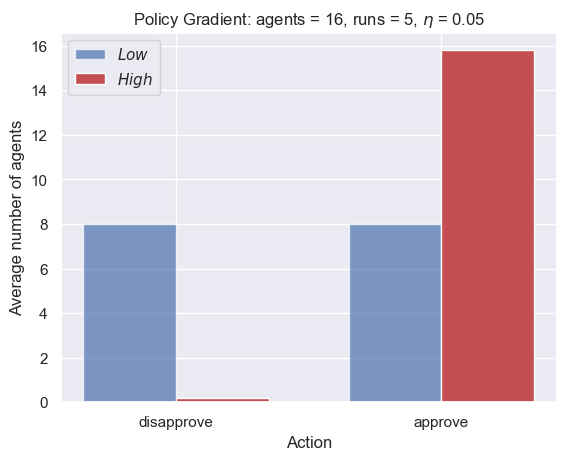}
    \caption{}
    \label{fig:a_pert}
  \end{subfigure}
  \begin{subfigure}{0.28\textwidth}
    \centering
    \includegraphics[width=\textwidth]{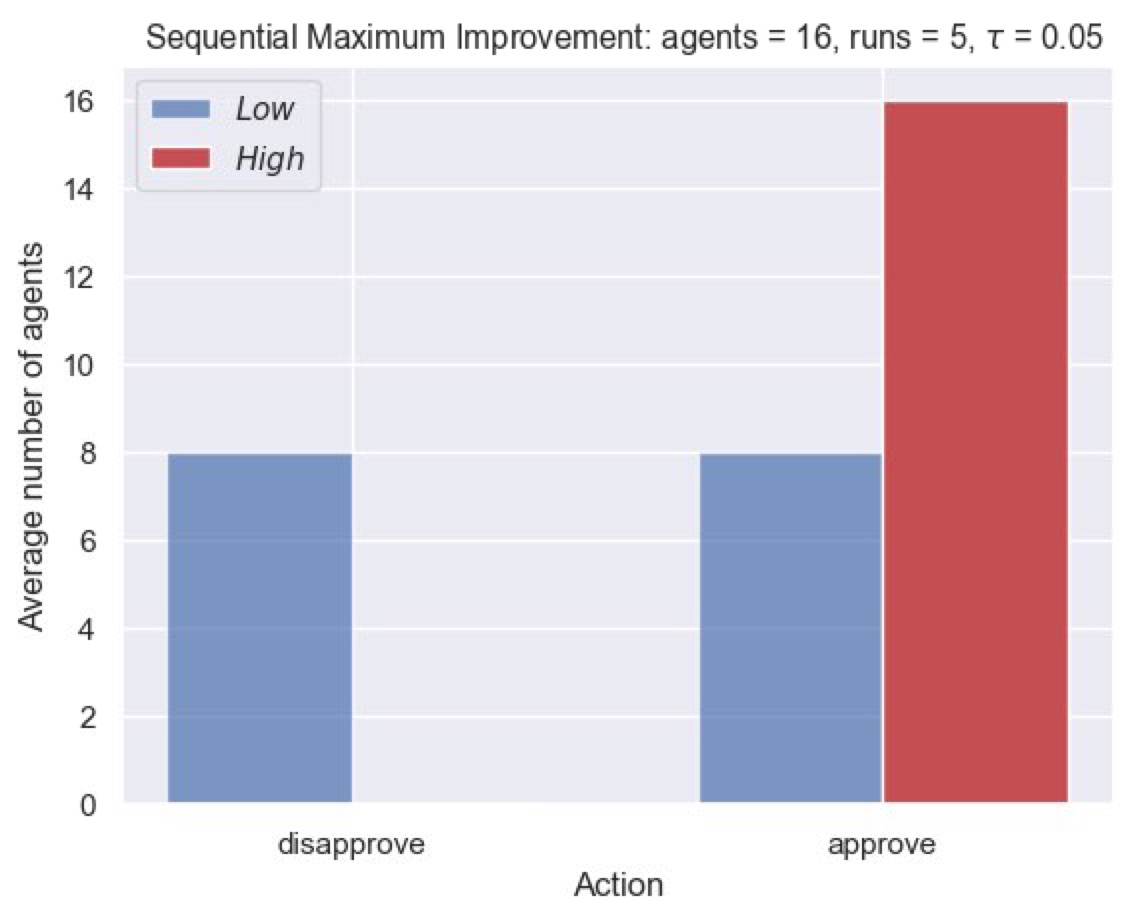}
    \caption{}
    \label{fig:b_pert}
  \end{subfigure}
  \begin{subfigure}{0.32\textwidth}
    \centering \includegraphics[width=\textwidth]{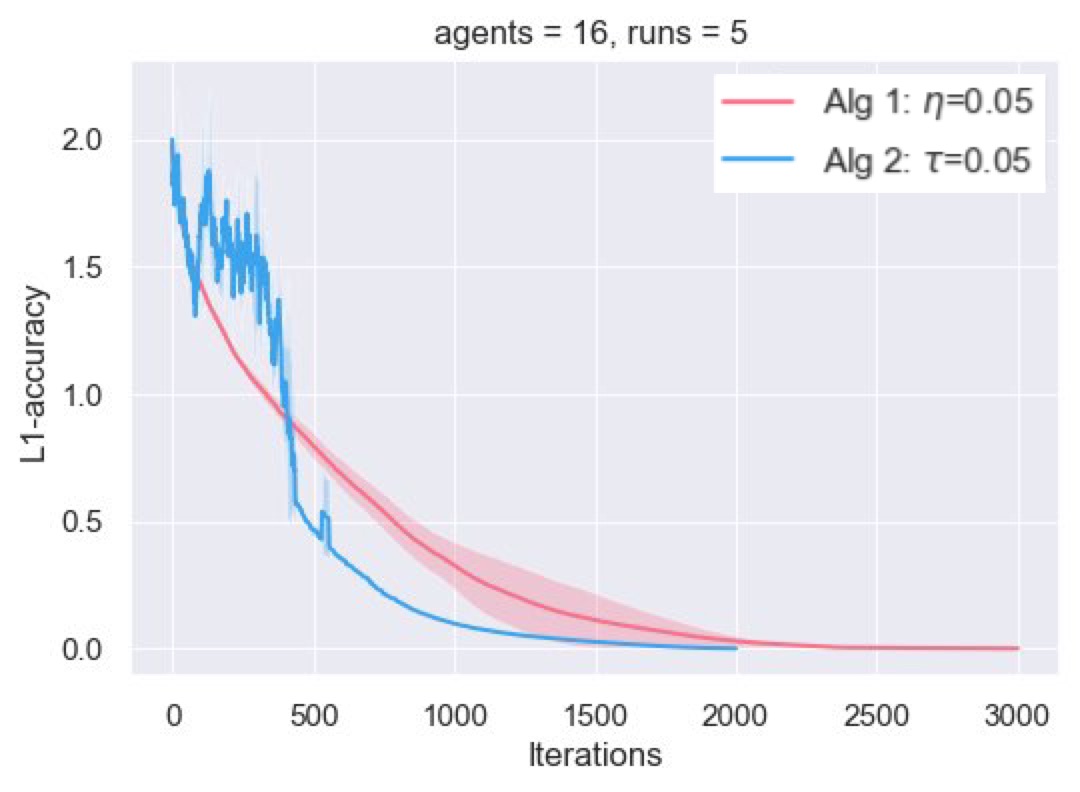}
    \caption{}
    \label{fig:c_pert}
  \end{subfigure}
        \end{minipage}
         \caption{\textbf{Perturbed Markov team game:} (a) and (b) are distributions of players taking actions in all states: (a) using Algorithm \ref{alg: PolicyGradient} with step-size $\eta = 0.05$; (b) using Algorithm \ref{alg: SequentialPureBRMaxPlayer} with regularizer $\tau_t =  0.9975^t \cdot  0.05 $. (c) is mean L1-accuracy with shaded region of one standard deviation over all runs.}
  \label{fig:pert}
\end{figure*}

\begin{figure}
    \centering
    \includegraphics[width=0.35\linewidth]{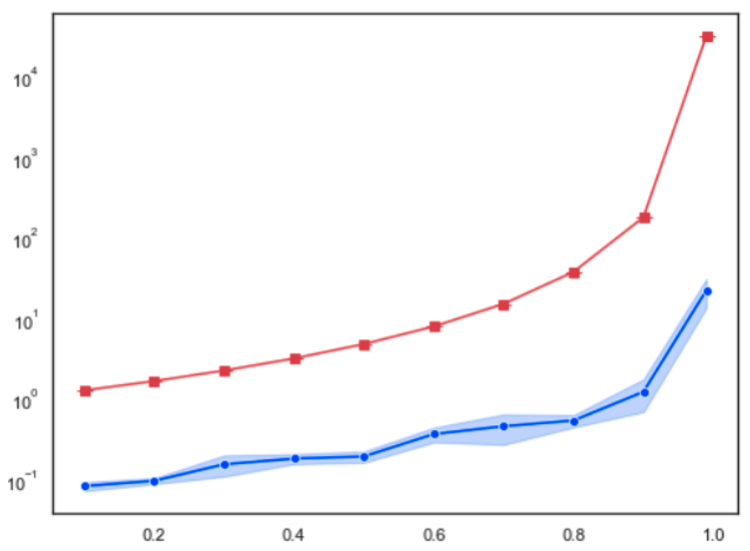}
    \caption{Variation of Nash regret with the discount factor for perturbed Markov team game with perturbation parameter \(\zeta = 0.1\). The red curve plots the function \(1/(1-\gamma)^{9/4}\) (as stated in Theorem 6.1) and the blue shaded region show the Nash regret computed through 10 rounds of experiments with random initialization. Note that the scale on y-axis is in log.}
    \label{fig:regret_bound_loose}
\end{figure}

\section{Conclusion}
We propose a new framework called Markov \(\mnpg\)-potential games to study Markov games, generalizing the framework of Markov potential games. We analytically compute upper bounds on \(\mnpg\) for Markov congestion games and perturbed Markov team games. We also present a semi-infinite linear programming approach to compute an upper bound on \(\mnpg\) for general discrete-time Markov games. This framework is used to design model-based MARL algorithms for  Markov games in discrete-time setting, {along with associated regret bounds.  However, a key limitation is that these regret bounds may not be tight, potentially restricting their applicability to Markov games with small values of $\alpha$.
An interesting direction for future research is therefore to develop sharp regret bounds. 
Another research direction is to design game structures by imposing specific rules or assumptions on the hyperparameters that lead to a theoretically small \(\alpha\), thereby improving tractability and performance guarantees. Additionally, one may develop (model-free) algorithms with tight upper and lower regret bounds.
It is worth mentioning some recent progress in \cite{guo2024alpha}, which has developed analytical tools to study general classes of continuous-time 
$\alpha$-potential games with arbitrary state and action spaces.  A natural next step is to investigate whether learning algorithms can be effectively applied to these more general settings.
}

\appendix
\section{Proofs in Section \ref{ssec: ExamplesNearMPG}}


\subsection{Proof of Proposition \ref{prop: exampleMPG}}
Let $\Phi$ be a potential function of MPG $\game$. 
Using Definition \ref{def: PotentialGame}, it suffices to show  $\Phi \in \mathcal{F}^\game$.
First,  we claim that for every \(s\in S, \pi,\pi'\in\Pi,\) 
\begin{align}\label{eqn:boundforPhi_minus_Phi}
    | \Phi(s, \pi) - \Phi(s, \pi') | \leq 
    \sum_{i=1}^N |V_i(s,\tilde{\pi}^{(i)})-V_i(s,\tilde{\pi}^{(i+1)})|,
\end{align}
where for any \(i\in \playerSet\), \(\tilde{\pi}^{(i)} = (\pi_1',\pi_2',..\pi_{i-1}',\pi_i,\pi_{i+1},...,\pi_N)\) with the understanding that \(\tilde{\pi}^{(1)} = \pi\) and \(\tilde{\pi}^{(N+1)} = \pi'\). 
To prove this claim, note that 
\begin{align*}
    | \Phi(s, \pi) - \Phi(s, \pi') |
    &= \bigg|\sum_{i=1}^{N} \Phi(s, \tilde{\pi}^{(i)}) - \Phi(s, \tilde{\pi}^{(i+1)}) \bigg| 
   \leq \sum_{i=1}^{N}\big| V_i(s, \tilde{\pi}^{(i)}) - V_i(s, \tilde{\pi}^{(i+1)}) \big|,
\end{align*}
which follows from Definition \ref{def: PotentialGame} as \(\tilde{\pi}^{(i)}\) and \(\tilde{\pi}^{(i+1)}\) only differ at player \(i\)'s policy. 
By \eqref{eqn:boundforPhi_minus_Phi}, for any $s\in S$, $\pi,\pi' \in \Pi$,
$$
|\Phi(s,\pi) -\Phi(s,\pi') | \leq 2N \max_{i\in\playerSet} \| V_i\|_{\infty} \leq \frac{2N}{1-\gamma}\max_{i\in\playerSet}\| u_i\|_{\infty}. 
$$
Without loss of generality, we have $\min_{\pi\in\Pi}\Phi(s,\pi) = 0$ for every \(s\in S\). Therefore,
$ \|\Phi \|_{\infty} \leq \frac{2N}{1-\gamma}\max_{i\in\playerSet}\| u_i\|_{\infty}.$

To show that \(\Phi\) lies in a uniformly equi-continuous set \(\mathcal{F}_{\mathcal{G}}\), we next show that \(\Phi\) is uniformly continuous. 
Note that for each \( s \in S \) and \( i \in I_N \), \( V_i(s, \cdot) : \Pi \rightarrow \mathbb{R} \) is a continuous function \cite[Lemma 2.10]{yongacoglu2023satisficing}.  Given that \( \Pi \) is compact and \(|S|<\infty\), for every \(\epsilon>0\) there exists \(\bar{\delta}(\epsilon)>0\) such that \(\underset{i\in I_N, s\in S}{\max}
|V_{i}(s, \pi) - V_i(s,\pi') | \leq {\epsilon }/N\) for any $\pi,\pi'\in \Pi$ satisfying
$\md(\pi, \pi') \leq \bar{\delta}(\epsilon)$. Consequently, from \eqref{eqn:boundforPhi_minus_Phi}, we conclude that for any \(\epsilon>0\), \(| \Phi(s, \pi) - \Phi(s, \pi') | \leq \epsilon\) for any  $\pi,\pi'\in \Pi$ satisfying
$\md(\pi, \pi') \leq  \bar{\delta}(\epsilon).$

\subsection{Proof of Proposition \ref{prop:mcg}}
The proof of Proposition \ref{prop:mcg} relies on the following lemma. 
\begin{lemma}\label{lem: CongestionGameTransitionDifference}
If there exists some $\zeta>0$ such that  for all \(s,s'\in S\), \( |P(s'|s,w)-P(s'|s,w')| \leq \zeta  \|w-w'\|_{1}\). Then for any \(i\in \playerSet, \pi_i,\pi_i'\in \Pi_i, \pi_{-i}\in \Pi_{-i}\),
   \begin{align}\label{eq: P_diff_inf_norm}
    \|P^{\pi_i,\pi_{-i}}-P^{\pi_i',\pi_{-i}}\|_\infty \leq {2\zeta}{ |S|\max_{a_i\in A_i}|a_i|}/{N}.
   \end{align}
\end{lemma}

\begin{proof}
For any \(i \in \playerSet\), \(\pi \in \Pi\), \(\pi_i' \in \Pi_i\), and \(s, s' \in S\),
\begin{align}
    &P^{\pi_i,\pi_{-i}}(s'|s) - P^{\pi_i',\pi_{-i}}(s'|s) \notag \\
    &= \mathbb{E}_{\substack{a_{-i} \sim \pi_{-i} \\ a_i \sim \pi_i}} \Big[ P(s'|s, w(a_i, a_{-i})) - P(s'|s, w(a_i, a_{-i})) \Big] \notag \\
    &\leq \mathbb{E}_{a_{-i} \sim \pi_{-i}} \Big[ P(s'|s, w(\bar{a}_i, a_{-i})) - P(s'|s, w(\underline{a}_i, a_{-i})) \Big], \label{eq: InitialTransitionDiff}
\end{align}
where the first equation is due to the structure of transition function, \(\bar{a}_i \in {\arg \max}_{a_i \in A_i} P(s'|s, w(a_i, a_{-i}))\), and \(\underline{a}_i \in {\arg \min}_{a_i \in A_i} P(s'|s, w(a_i, a_{-i}))\).
By \eqref{eq: InitialTransitionDiff} and the Lipschitz property of the transition matrix in Lemma \ref{lem: CongestionGameTransitionDifference},
\begin{align*}
    &\sum_{s'\in S}|P^{\pi_{i},\pi_{-i}}(s'|s) - P^{\pi_{i}',\pi_{-i}}(s'|s)| \\ 
    &\stackrel{(a)}{\leq}\frac{\zeta|S|}{ N }    \underset{a_{-i} \sim \pi_{-i}}{\avg} \left[\sum_{e\in E}|\mathbbm{1}(e\in \bar{a}_i)-\mathbbm{1}(e\in\underline{a}_i)| \right]\\
    &\leq \frac{2\zeta |S|\max_{a_i\in A_i}|a_i|}{ N }, \quad \forall \ s\in S,
\end{align*}
where \((a)\) follows by   \eqref{eq:w}. 
\end{proof}

\paragraph{Proof of Proposition \ref{prop:mcg}.}
    Recall that for any \(s\in S\), the stage game is a potential game with a potential function \(\varphi(s, a) = 1/ N \sum_{e\in E}\sum_{j=1}^{w_e(a)  N }c_e\left(s, j/ N \right)\).
Under this notation, we can equivalently write \eqref{eq:potential} as
\begin{align}\label{eq: potentialBellmanForm}
    \Psi(s,\pi)  
    &= \varphi(s,\pi) + \discount \sum_{s'\in S}P^\pi(s'|s)\Psi(s',\pi).
\end{align}

For the rest of the proof, fix arbitrary \(\pi_i,\pi_i'\in \Pi_i, \pi_{-i}\in \Pi_{-i}\) and denote \(\pi=(\pi_i,\pi_{-i}),\pi'=(\pi_i',\pi_{-i})\). By \eqref{eq: potentialBellmanForm}, 
\begin{align}\label{eq: PotentialDiff}
     \Psi(s,\pi)-\Psi(s,\pi') =\varphi(s,\pi) - \varphi(s,\pi')  + \discount\sum_{s'\in S}\lr{P^\pi(s'|s)\Psi(s',\pi) - P^{\pi'}(s'|s)\Psi(s',\pi')}. 
\end{align}
Additionally, recall that \( V_i(s,\pi) = u_i(s,\pi) + \discount \sum_{s'\in S}P^\pi(s'|s) V_{i}(s',\pi).\)
Consequently, 
\begin{align}\label{eq: ValueDiff}
    &V_i(s,\pi)- V_i(s,\pi') = u_i(s,\pi)- u_i(s,\pi')\\
    &\quad+ \discount \sum_{s'\in S}\left(P^\pi(s'|s) V_{i}(s',\pi)-P^{\pi'}(s'|s) V_{i}(s',\pi')\right).  \notag
\end{align}
Subtracting \eqref{eq: PotentialDiff} from \eqref{eq: ValueDiff}, we obtain
\begin{align*}
    V_i(s,\pi)&- V_i(s,\pi') - \lr{\Psi(s,\pi)-\Psi(s,\pi') } \\
    {=}  & \discount \sum_{s'\in S}P^\pi(s'|s) \lr{V_{i}(s',\pi)-\Psi(s',\pi)} -\gamma\sum_{s'\in S}P^{\pi'}(s'|s) \lr{V_{i}(s',\pi')-\Psi(s',\pi')}  \\ 
    {=}& \discount \sum_{s'\in S}P^\pi(s'|s) \lr{V_{i}(s',\pi)- V_i(s',\pi') +  \Psi(s',\pi') - \Psi(s',\pi) } \\ 
    &\quad -\gamma \sum_{s'\in S}\lr{P^{\pi'}(s'|s)-P^\pi(s'|s)} \lr{V_{i}(s',\pi')-\Psi(s',\pi')}.
\end{align*}
Thus, 
\begin{align}
    &\max_{s\in S }|V_i(s,\pi)- V_i(s,\pi') -(\Psi(s,\pi)-\Psi(s,\pi')) | \label{eqn: max Vdiff}  \\ 
   &\leq \discount\max_{s\in S }|V_i(s,\pi)- V_i(s,\pi') -\lr{\Psi(s,\pi)-\Psi(s,\pi')} | \notag \\ 
    &+ \discount \max_{s'\in S}|\Psi(s',\pi')-V_{i}(s',\pi')| \max_{s\in S}\sum_{s'\in S}\left|P^\pi(s'|s)-P^{\pi'}(s'|s)\right|. \notag
\end{align}
Rearranging terms leads to
\begin{align}\label{eq: MCGPrePlugPDiff}
    \eqref{eqn: max Vdiff} 
    &\leq  \frac{\discount}{1-\discount}\max_{s'\in S}|\Psi(s',\pi')-V_{i}(s',\pi')|\|P^{\pi}-P^{\pi'}\|_{\infty}
    \notag \\
    &\leq  \frac{2\discount\zeta |S|\max_{a_i\in A_i}|a_i|}{(1-\discount) N }\max_{s'\in S}|\Psi(s',\pi')-V_{i}(s',\pi')|.
\end{align}
where the last inequality follows from Lemma \ref{lem: CongestionGameTransitionDifference}.
Finally, since
$$
\begin{aligned}
    u_i(s^k,a^k) &={\sum_{e\in E}c_{e}(s^k,w_e^k) \mathbbm{1}(e \in a_i^k)} 
    \leq  \sum_{e\in E}c_{e}(s^k,w_e^k)\leq \varphi(s^k,a^k),
\end{aligned}
    $$ 
    then for any \(s'\in S\),
\begin{align*}
   |\Psi(s',\pi')-V_{i}(s',\pi')| &\leq   \avg_{\pi'}\left[\sum_{k=0}^{\infty}\discount^k \left|\varphi(s^k,a^k) -  u_i(s^k,a^k)\right| \right] \\
   &{\leq} \left|\avg_{\pi'}\left[\sum_{k=0}^{\infty}\discount^k \varphi(s^k,a^k)\right]\right| \leq \sup_{s,\pi}\Psi(s,\pi).
\end{align*}
Plugging the above inequality into \eqref{eq: MCGPrePlugPDiff} finishes the proof.
\endproof

\subsection{Proof of Proposition \ref{prop: PerturbedMarkovTeamGame}}
Throughout the proof, let us fix arbitrary \(i\in \playerSet, \pi_i,\pi_i'\in \Pi_i, \pi_{-i}\in \Pi_{-i}\), and define \(\pi=(\pi_i,\pi_{-i}),\pi'=(\pi_i',\pi_{-i})\). 
We show that for every \(i\in \playerSet, \pi_i, \pi_i'\in \Pi_i, \pi_{-i}\in \Pi_{-i}\),  
$$
\begin{aligned}
     \underset{{s\in S }}{\max}|&V_i(s,\pi)- V_i(s,\pi') -(\Psi(s,\pi)-\Psi(s,\pi')) | \leq \frac{2\kappa}{(1-\discount)^2},
\end{aligned}
   $$
where \(\Psi(s, \pi)\coloneqq \avg_\pi\left[\sum_{k=0}^{\infty}\discount^k r(s^k, a^k)|s^0=s\right]\). Note that 
\begin{align}\label{eq: DefPhiMPTG}
\Psi(s, \pi)
&= r(s,\pi)+\discount \sum_{s'\in S}P^\pi(s'|s)\Psi(s',\pi).
\end{align}
By \eqref{eq: DefPhiMPTG}, for any \(s\in S\), 
\begin{equation}
    \begin{aligned}\label{eq: PotentialDiffPTMG}
    &\Psi(s,\pi)-\Psi(s,\pi') = r(s,\pi) - r(s,\pi')+ \discount\sum_{s'\in S}\lr{P^\pi(s'|s)\Psi(s',\pi) - P^{\pi'}(s'|s)\Psi(s',\pi')  }.
    \end{aligned}
\end{equation}
Similarly, for any \(s\in S\), 
\begin{equation}
    \begin{aligned}\label{eq: ValueDiffPTMG}
    &V_i(s,\pi)- V_i(s,\pi')= u_i(s,\pi)- u_i(s,\pi') + \discount \sum_{s'\in S}P^\pi(s'|s) V_{i}(s',\pi)-P^{\pi'}(s'|s) V_{i}(s',\pi'). 
\end{aligned}
\end{equation}
Consequently, 
\begin{align*}
    &V_i(s,\pi)- V_i(s,\pi') - \lr{\Psi(s,\pi)-\Psi(s,\pi') } \\ 
    &{=} u_i(s,\pi)- u_i(s,\pi') - \lr{r(s,\pi) - r(s,\pi') } \\
    &-\discount\sum_{s'\in S}\lr{P^{\pi'}(s'|s)-P^\pi(s'|s)} \lr{V_{i}(s',\pi')-\Psi(s',\pi')} 
   \\
    &+ \discount \sum_{s'\in S}P^\pi(s'|s) \lr{V_{i}(s',\pi)- V_i(s',\pi') +  \Psi(s',\pi') - \Psi(s',\pi) }.
\end{align*}
Since $|u_i(s,\pi)- u_i(s,\pi') - (r(s,\pi) - r(s,\pi') )| \leq 2 \|\xi_i\|_{\infty} \leq 2\kappa,$ then
\begin{align}\label{eq: pMTGVDiff}
    &\max_{s\in S }|V_i(s,\pi)- V_i(s,\pi') -\lr{\Psi(s,\pi)-\Psi(s,\pi')} |   \\ 
   &{\leq}  2\kappa  + 2\discount\max_{s'\in S}|\Psi(s',\pi')-V_{i}(s',\pi')| \notag
    \\ 
    &\quad + \discount\max_{s\in S }|V_i(s,\pi)- V_i(s,\pi') -\lr{\Psi(s,\pi)-\Psi(s,\pi')} |.\notag
\end{align}
Rearranging terms in above inequality, we obtain 
\begin{align}\label{eq: Final1}
    &\eqref{eq: pMTGVDiff} \leq  
    \frac{2\kappa}{1-\discount} + \frac{2\discount}{1-\discount}\max_{s'\in S}|\Psi(s',\pi')-V_{i}(s',\pi')|. 
\end{align}
Note that \(
    |\Psi(s',\pi')-V_{i}(s',\pi')| = |\sum_{k=0}^{\infty}\discount^k\xi_i(s,\pi'(s^k))| \leq {\kappa}/({1-\discount}). 
\)
Plugging this inequality into \eqref{eq: Final1} completes the proof.
\endproof

\section{Proofs in Section \ref{ssec: FindClosestMPG}}
\subsection{Proof of Proposition \ref{lemma: Phi_constructed_MDP}}
To prove Proposition \ref{lemma: Phi_constructed_MDP}, we first need the following lemma.
\begin{lemma}
[Lemma B.1 in \cite{yongacoglu2023satisficing}]
\label{lemma: trans_prob_continuous}
    Fix $i \in \playerSet$ and $K \in \mathbb{N}$. For any $s \in S$ and $\omega = \left(\tilde{s}^k, \tilde{{a}}^k\right)_{k=0}^K \in$ $(S \times A)^{K+1}$, the mapping \(
\Pi\ni \pi\mapsto \mathbb{E}_{\pi}\left[ \mathbbm{1}\left((s^k, a^k)_{k=0}^{K} = \omega\right) \mid s^0=s\right] 
\)
is continuous.
\end{lemma}

\paragraph{Proof of Proposition \ref{lemma: Phi_constructed_MDP}.}
Fix $\epsilon > 0$ and define 
 $ M \coloneqq N \max_{i\in\playerSet} \| u_i\|_{\infty}$. Choose $K \in \mathbb{N}$ large enough that
$
\frac{\discount^K \cdot M}{1-\discount}  <\frac{\epsilon}{4}
$
and $\tilde\epsilon \coloneqq \frac{(1-\gamma)\epsilon}{2M|S|^{K+1}|A|^{K+1}}.$
Since $\Pi$ is compact and $S\times A$ is finite, Lemma \ref{lemma: trans_prob_continuous} ensures that
there exists ${\delta}(\epsilon)$ such that 
for any $\pi,\pi'\in\Pi$ 
with $\md(\pi,\pi') \leq {\delta}(\epsilon)$, and $\omega \in (S\times A)^{K+1}, s \in S,$
\begin{align}\label{eqn: P_minus_P}
    \bigg| \mathbb{E}_{\pi}&\left[ \mathbbm{1}\left((s^k, a^k)_{k=0}^{K} = \omega\right) \mid s^0=s\right]   - \mathbb{E}_{\pi'}\left[ \mathbbm{1}\left((s^k, a^k)_{k=0}^{K} = \omega\right) \mid s^0=s\right]  \bigg| \leq \tilde\epsilon.
\end{align}
From \eqref{eqn: F tilde}, we note that for any   $\Psi  \in \Tilde{\mathcal{F}}^\game $, there exists \(\phi:S\times A\ra \mathbb{R}\) such that for any $\pi, \pi'\in \Pi, s\in S, $
    \begin{align}
 &\left|\Psi(s,{\pi})-\Psi(s,{\pi'})\right| \leq  \Bigg|\mathbb{E}_{{\pi}}\left[\sum_{k=0}^K\gamma^k \phi \left(s^k, a^k\right) \mid s_0=s\right]  -\mathbb{E}_{{{\pi'}}}\left[\sum_{k=0}^K\gamma^k \phi \left(s^k, a^k\right) \mid s_0=s\right]\Bigg|+\frac{\epsilon}{2} . \label{eqn:first_eqn_psi_minus_psi}
\end{align}
Define a function $\varphi:(S \times A)^{K+1} \rightarrow \mathbb{R}$ such that for every $\left(\tilde{s}^k, \tilde{{a}}^k\right)_{k=0}^K \in(S \times A)^{K+1}$,
$
\varphi\left(\tilde{s}^0, \tilde{{a}}^0, \cdots, \tilde{s}^K, \tilde{{a}}^K\right):=\sum_{k=0}^K\gamma^k \phi \left(\tilde{s}^k, \tilde{{a}}^k\right).
$
Thus, for any ${\pi} \in \Pi$,
\begin{equation*}
\begin{aligned}
        &\mathbb{E}_{{\pi}}\left[\sum_{k=0}^K\gamma^k \phi\left(s^k, a^k\right) \mid s^0=s\right]=\sum_{\omega \in( S  \times   A )^{K+1}} \varphi(\omega) \mathbb{E}_{\pi}\left[\mathbbm{1}\left(\left(s^k,   a^k\right)_{t=0}^K=\omega \right)\bigg| s^0=s\right].
 \end{aligned}
\end{equation*}
 Thus, by applying the above equation and \eqref{eqn: P_minus_P} to \eqref{eqn:first_eqn_psi_minus_psi}, we obtain that for any $s\in S, \pi,\pi'\in \Pi$ satisfying $\md(\pi,\pi')\leq \delta(\epsilon)$,
\begin{align*}
    \left|\Psi(s,{\pi})-\Psi(s,{\pi'})\right|   &\leq {\| \varphi\|_{\infty} |S|^{K+1} |A|^{K+1} \tilde\epsilon } + \frac{\epsilon}{2}\\
    &\leq \frac{M |S|^{K+1} |A|^{K+1} \tilde\epsilon}{1-\gamma} + \frac{\epsilon}{2} \leq \epsilon.
\end{align*}
Since we chose arbitrary $\Psi\in \Tilde{\mathcal{F}}^\game $, and $\delta$ is independent of the choice of $\Psi$, then
$\Tilde{\mathcal{F}}^\game $ is equi-continuous. 
Thus,  $\tilde{\mathcal{F}}^\game\subseteq \mathcal{F^\game}.$

\section{Proofs in Section \ref{sssec: AlgPG}}

\subsection{Proof of Lemma \ref{lemma: policy improve}}
    
To prove Lemma \ref{lemma: policy improve}, we define $\pi_{i \sim j} \coloneqq \left\{\pi_k\right\}_{k=i+1}^{j-1}$ as the joint policy for players from $i+1$ to $j-1$; $\pi_{<i} \coloneqq \{ \pi_k\}_{k=1}^{i-1}$, and $\pi_{>j} \coloneqq  \{ \pi_k\}_{k=j+1}^{N}$ are defined similarly. Next, we recall a useful result from \cite{ding2022independent}.
\begin{lemma}[Lemma 2 in \cite{ding2022independent}]\label{lemma: multi func diff}
    For any function $f: \Pi \rightarrow \mathbb{R}$, and any two policies $\pi, \pi^{\prime} \in \Pi$,
    \begin{align}\label{eqn: aux f}
    f({\pi^{\prime}})-f(\pi)&= \sum_{i=1}^{ N }\left(f({\pi_i^{\prime}, \pi_{-i}})-f(\pi)\right)\notag \\
    & \quad +\sum_{i=1}^{ N } \sum_{j=i+1}^{ N }\Big(f({\pi_{<i, i \sim j}, \pi_{>j}^{\prime}, \pi_i^{\prime}, \pi_j^{\prime}})-f({\pi_{<i, i \sim j}, \pi_{>j}^{\prime}, \pi_i,    \pi_j^{\prime}})\notag \\
    &\quad
    -f({\pi_{<i, i \sim j}, \pi_{>j}^{\prime}, \pi_i^{\prime}, \pi_j})
    +f({\pi_{<i, i \sim j}, \pi_{>j}^{\prime}, \pi_i, \pi_j})\Big).
    \end{align}
\end{lemma} 
Next, we state a result that lower bounds the improvement in value function of each player in each step of Algorithm \ref{alg: PolicyGradient}.
{\begin{lemma}\label{lemma: value policy improve} Consider a Markov game \(\mathcal{G}\) with initial state distribution $\nu$, let $\pi^{(t+1)}$ and $\pi^{(t)}$ be consecutive policies in Algorithm \ref{alg: PolicyGradient}. Then we have,
\begin{equation*}
    \begin{aligned}
     \text{(i)}\,& V_i(\nu, {\pi^{(t+1)}})-V_i(\nu, {\pi^{(t)}}) \geq -\frac{4 \eta^2 \bar{A}^2N^2}{(1-\discount)^5} +\frac{1}{2 \eta(1-\discount)}  \cdot \sum_{i\in \playerSet ,s\in S} d_\nu^{\pi_i^{(t+1)}, \pi_{-i}^{(t)}}(s)\left\|\pi_i^{(t+1)}(s) -\pi_i^{(t)}(s)\right\|^2;
 \end{aligned}
\end{equation*}
\begin{equation*}
    \begin{aligned}
      \text{(ii)}\, &V_i(\nu, {\pi^{(t+1)}})-V_i(\nu, {\pi^{(t)}}) \geq \frac{1}{2 \eta(1-\discount)} \left(1-\frac{4 \eta \kappa_\nu^3 \bar{A} N}{(1-\discount)^4}\right) 
      \cdot \sum_{i=1}^{N} \sum_{s\in S} d_\nu^{\pi_i^{(t+1)}, \pi_{-i}^{(t)}}(s) \left\|\pi_i^{(t+1)}(s)-\pi_i^{(t)}(s)\right\|^2.
 \end{aligned}
 \end{equation*}
\end{lemma}
\begin{proof}
This result directly follows from \cite[Lemma 3]{ding2022independent}. Specifically, the proof of \cite[Lemma 3]{ding2022independent} is established by lower-bounding the difference $\Phi(\nu, \pi^{(t+1)}) - \Phi(\nu, \pi^{(t)})$ for a Markov potential game with potential function $\Phi$. At its core, the proof relies on the key property of Markov potential games, which allows the difference in potential functions to be expressed as the difference in value functions for each player. The remainder of the proof focuses on lower-bounding the difference in value functions at each step of the policy update process in Algorithm \ref{alg: PolicyGradient}, which is precisely what we require. We omit details due to space constraints. 
\end{proof}
}
\paragraph{Proof of Lemma \ref{lemma: policy improve}.}
For ease of exposition, let $\pi^{\prime}=\pi^{(t+1)}$ and $\pi=\pi^{(t)}$. 
By Definition \ref{def: alpha_param}, $|V_i(\nu,\pi_i',\pi_{-i}) - V_i(\nu,\pi_i,\pi_{-i}) - (\Phi(\nu,\pi_i',\pi_{-i}) -\Phi(\nu,\pi_i,\pi_{-i})) | \leq \mnpg$ for any $\nu, i\in \playerSet, \pi_i, \pi_i' \in \Pi_i$ and $\pi_{-i} \in \Pi_{-i}$.
Apply Lemma \ref{lemma: multi func diff} with $f(\cdot)=V_i( \nu, \cdot) - \Phi( \nu, \cdot )$ respectively. Since each term in \eqref{eqn: aux f} only differs in one player's policy, 
we obtain
{$$
\begin{aligned}
    &|V_i (\nu,\pi') - V_i(\nu, \pi) - (\Phi(\nu,\pi')-\Phi(\nu,\pi'))| \leq \sum_{i=1}^N \alpha + \sum_{i=1}^N \sum_{j=i+1}^N \alpha \leq N^2 \alpha.
\end{aligned}
$$}
The proof follows by the above inequality and Lemma \ref{lemma: value policy improve}.

\section*{Proofs in Section \ref{sssec: AlgBR}}
  
\subsection{Proof of Lemma \ref{lem: performanceDiffPG}}
Fix arbitrary \(i\in \playerSet, \mu\in \mathcal{P}(S), \pi_i, \pi_i'\in \Pi_i, \pi_{-i}\in \Pi_{-i}\). We define \(\policy=(\policy_i,\policy_{-i}), \policy'=(\policy_i',\policy_{-i})\in \Pi\). Note that 

  \begin{align}
    &\tilde{V}_i(\mu,\policy)-\tilde{V}_i(\mu,\policy') =\avg_{\pi}\bigg[\sum_{k=0}^{\infty}\discount^k \bigg( \stagePayoff_i(s^k,a^k)-\tau \sum_{j\in \playerSet}\nu_j(s^k,\policy_j) - \tilde{V}_i(s^k,\policy')+\tilde{V}_i(s^k,\policy') \bigg) \bigg] -\tilde{V}_i(\mu,\policy')
     \notag\\ 
    &=\avg_{\pi}\bigg[ \sum_{k=0}^{\infty}\discount^k \bigg( \stagePayoff_i(s^k,a^k) -\tau \sum_{j\in \playerSet}\nu_j(s^k,\policy_j) - \tilde{V}_i(s^k,\policy')  \bigg) \bigg] + \avg_\pi \ls{\sum_{k=1}^{\infty}\gamma ^{k}\tilde{V}_i(s^k,\policy')}.\label{eqn:tmp_smooth_perf_diff}
    \end{align}
     Note that 
$$ \avg_{\pi}\big[\sum_{k=1}^{\infty}\gamma ^{k}\tilde{V}_i(s^k,\policy')\big] = 
\gamma\avg_{\pi}\ls{\sum_{k=0}^{\infty} \gamma ^{k}\tilde{V}_i(s^{k+1},\policy') },$$
thus,
    \begin{align}
\eqref{eqn:tmp_smooth_perf_diff}
&=\mathbb{E}_\pi\Big[\sum_{ k = 0 } ^ { \infty } \gamma^{ k } \big(u_i\left(s^k, a^k\right) -\tau \sum_{j \in I_N} \nu_j\left(s^k, \pi_j\right) -\tilde{V}_i\left(s^k, \pi^{\prime}\right)+\gamma \tilde{V}_i\left(s^{k+1}, \pi^{\prime}\right)\big)\Big] \notag\\
&=\avg_{\pi}\bigg[\sum_{k=0}^{\infty}\discount^k \bigg(  \stagePayoff_i(s^k,a^k)-\tau\sum_{j\in \playerSet}\nu_j(s^k,\policy_j') { +\discount \sum_{s'\in S}P(s'|s^k,a^k)\tilde{V}_i(s',\policy')}- \tilde{V}_i(s^k,\policy') \notag\\
    &\quad\quad \quad +\tau\sum_{j\in \playerSet}\nu_j(s^k,\policy_j')-\tau\sum_{j\in \playerSet}\nu_j(s^k,\policy_j)  \bigg) \bigg].  \label{eqn:tildeV_minus_tildeV}
\end{align}
We can continue the above calculations by applying smoothed $Q$-function and noting that \(\pi_j'=\pi_j\) for all \(j\neq i\) and { \(\tilde{V}_i(s',\pi') = \pi_i'(s')^\top \tilde{\bQ}^{\pi'}_i(s')\),} 
\begin{align*}
    \eqref{eqn:tildeV_minus_tildeV}
    &= \avg_{\pi_i}\bigg[ \sum_{k=0}^{\infty}\discount^k \bigg( { \tilde{Q}_i^{\policy'}(s^k,a^k_i)}- \tilde{V}_i(s^k,\policy')+\tau\sum_{j\in \playerSet}\nu_j(s^k,\policy_j')-\tau \sum_{j\in \playerSet}\nu_j(s^k,\policy_j)  \bigg) \bigg]\\
    &{=} \frac{1}{1-\discount}\sum_{s'\in S}d^{\pi}_{\mu}(s') {  \Big(\pi_i(s')-\pi_i'(s')\Big)^\top\tilde{\bQ}^{\policy'}_i(s')} +\tau\ti(s',\policy_i')-\tau\ti(s',\policy_i)\Big).
\end{align*}

\subsection{Proof of Lemma \ref{lem: PerturbGameInitialGameDiff}}
From the definition of smoothed infinite horizon utility \eqref{eqn: smoothed value function}, we note that 
 for every \(i\in \playerSet, \pi_i\in\Pi_i,\pi_{-i}\in\Pi_{-i},s\in S\),
 \begin{equation}\label{eqn: V tilde and V}
 \begin{aligned}
     &\tilde{V}_i(s,\pi_i,\pi_{-i})  = V_i(s,\pi_i,\pi_{-i}) -\tau\avg_{\pi}\ls{\sum_{k=0}^{\infty}\gamma^k\sum_{j\in \playerSet}\nu_j(s^k,\pi_j) | s_0 = s}.
 \end{aligned}
 \end{equation}
Using \eqref{eqn: V tilde and V}, it holds that for any $\mu \in \mathcal{P}(S) $ and $\pi \in \Pi,$
\begin{align}
  &|\tilde{\vFunc}_i(\mu,\pi) -V_i(\mu,\pi)| = \tau \left| \avg_{\mu,\pi}\ls{\sum_{k=0}^{\infty}\discount^k\sum_{j\in \playerSet}\nu_j(s^k,{\policy}_j)} \right| \notag \\
  &\leq \frac{\tau N  \max_{s,\pi_i}\nu_i(s,\pi_i)}{1-\gamma}= \frac{\tau  N  \log(\bar{A})}{1-\gamma}.\label{eqn: bound of regu}
\end{align}
{The desired result follows from triangle inequality and  \eqref{eqn: bound of regu}.}

\subsection{Proof of Lemma \ref{lemma:technical}}
The proof of Lemma \ref{lemma:technical}
requires the following technical lemmas.
\begin{lemma}\label{lem: perturbedGameNMPG}
    If \(\game\) is a Markov \(\mnpg\)-potential game with \(\Phi\) as its \(\mnpg\)-potential function, 
    then for any \(s\in S, i\in \playerSet, \pi_i',\pi_i\in \Pi_i, \pi_{-i}\in \Pi_{-i}\),
    \(
\Big|(\tilde{\Psi}(s, \pi_i^{\prime}, \pi_{-i})-\tilde{\Psi}\left(s, \pi_i, \pi_{-i}\right)) - (\tilde{V}_i(s, \pi_i^{\prime}, \pi_{-i})-\tilde{V}_i(s, \pi_i, \pi_{-i}))\Big|\leq \mnpg,
\)
    where
    \[
        \tilde{\Psi}(s,\pi) := \Phi(s,\pi) - \tau \avg_{\pi}[\sum_{j\in \playerSet}\sum_{k=0}^{\infty}\gamma^k \nu_j(s^k,\pi_j) \mid s^0 = s].
    \]
\end{lemma}
\proof{To ease the notation, for function $f:S\times \Pi \ra \R$, we write $f(s,\cdot)$ as $f^s(\cdot)$. 
By \eqref{eqn: V tilde and V} and the definition of $\tilde\Psi$ in Lemma \ref{lemma:technical}, we have
 for all \(s\in S, i\in \playerSet, \pi_i',\pi_i\in\Pi_i, \pi_{-i}\in \Pi_{-i}\),
\begin{align*}
   &|\tilde{\Psi}^s( \pi_i^{\prime}, \pi_{-i})-\tilde{\Psi}^s( \pi_i, \pi_{-i}) - (\tilde{V}_i^s( \pi_i^{\prime}, \pi_{-i})-\tilde{V}_i^s(\pi_i, \pi_{-i}))|\\ =& |{\Phi}^s( \pi_i^{\prime}, \pi_{-i})-{\Phi^s}( \pi_i, \pi_{-i})- (V_i^s(\pi_i^{\prime}, \pi_{-i})-V_i^s( \pi_i, \pi_{-i}))|,
\end{align*}
    which is bounded by $\mnpg$ using Definition \ref{def: mnpg}.
\endproof}

\begin{lemma}\label{lem: FirstOrderOptimality}
    For any \(i\in \playerSet, s\in S, \pi_i'\in \Pi_i, t\in [T]\), it hold that 
    \begin{align*}
    &\sum_{a_i\in A_i}{  \tilde{Q}_i^{(t)}(s,a_i)}\left(\text{BR}_i^{(t)}(a_i|s)-\pi_i'(a_i|s)   \right) \\
    &\geq \tau\sum_{a_i\in A_i}  \log\left(\text{BR}_i^{(t)}(a_i|s)\right)\left( \text{BR}_i^{(t)}(a_i|s)  -\pi_i'(a_i|s)  \right).
\end{align*}
\end{lemma}
\proof{Fix arbitrary \(i\in \playerSet, s\in S\), and \(t\in [T]\).
Next, note that the optimization problem in \eqref{eq: SmoothedBR} is a strongly concave optimization problem. By the first order conditions of constrained optimality,  for all $\pi_i'\in \Pi_i$,
    \[
    \lr{\tilde{\bQ}_i^{(t)}(s)-\tau \nabla_{\pi_i(s)} \nu_i(s,\text{BR}_i^{(t)}(s))}^\top (\text{BR}_i^{(t)}(s)-\pi_i'(s)  ) \geq 0. 
\]
Note that \(\nabla_{\pi_i(a_i|s)} \nu_i(s,\pi_i) = 1+\log(\pi_i(a_i|s)) \) for every \(a_i\in A_i\). Therefore, for every \(\pi_i'\in \Pi_i\), 
 \begin{align*}
    &\sum_{a_i\in A_i} \tilde{Q}^{(t)}_i(s,a_i)\left(\text{BR}_i^{(t)}(a_i|s)-\pi_i'(a_i|s)   \right)\\
    &\geq \tau\sum_{a_i\in A_i}  \lr{1+\log\left(\text{BR}_i^{(t)}(a_i|s) \right)}\left( \text{BR}_i^{(t)}(a_i|s)  -\pi_i'(a_i|s)\right).
\end{align*}
The result follows by noting that
{\(\sum_{a_i\in A_i}\text{BR}_i^{(t)}(a_i|s)=\sum_{a_i\in A_i}\pi_i'(a_i|s)=1\)}. 
\endproof}
\begin{lemma}\label{lem: strongConvexityEntropy}
    For any \(i\in \playerSet, s\in S, \pi_i,\pi_i'\in \Pi_i\),
    \begin{align*}
       & \nu_i(s,\pi_i) - \nu_i(s,\pi_i') \geq \frac{1}{2}\|\pi_i(s)-\pi_i'(s)\|^2  +\sum_{a_i\in A_i}\lr{\log(\pi_i'(a_i|s))}\lr{\pi_i(a_i|s)-\pi_i'(a_i|s)}.
    \end{align*}
\end{lemma}
\proof{
Fix arbitrary \(i\in \playerSet, s\in S\).     To prove the lemma, we first claim that the mapping \(\mathcal{P}(A_i)\ni\pi\mapsto \nu_i(s,\pi)\) is \(1\)-strongly convex. This can be observed by computing the Hessian, which is a \(\R^{A_i\times A_i}\) diagonal matrix with \((a_i,a_i)\) entry as \(1/\pi(a_i|s)\). Since \(\pi(a_i|s)\leq 1\), it follows that the diagonal entries of the Hessian matrix are all greater than $1$. Thus, \(\nu_i(s,\cdot)\) is \(1\)-strongly convex function. The result follows by noting that for any \(\kappa\)-strongly convex function \(f\), 
    \(
        f(y) \geq f(x) + \nabla f(x)^\top (y-x) + \frac{\kappa}{2}\|y-x\|^2. 
    \)
\endproof}

\begin{lemma}\label{lem: LogPolicy}
    For any \(i\in \playerSet, t\in [T], a\in A_{\imax}\), there exists \(0\leq t^\ast\leq t\) such that 
    \[
         \tau |\log(\pi_{\imax}^{(t)}(a|\smax))| \leq 2\|{ \tilde{\bQ}_{\imax}^{(t^\ast)}(\smax)}\|_\infty + \tau \log(|A_{\imax}|).
    \]
\end{lemma}
\proof{Recall that in Algorithm \ref{alg: SequentialPureBRMaxPlayer},
at any time step \(t\in [T]\), player \(\bar{i}^{(t)}\) updates her policy at time \(t+1\) in the state \(\bar{s}^{(t)}\), while policies for other players and other states remain unchanged. 
Fix arbitrary \(t\in [T]\). Let \(0\leq t^\ast \leq t\) be the latest time step when player \(\imax\) updated its policy in state \(\smax\) before time \(t\). Note that \(t^\ast = 0\) if \(t\) is the first time when player \(\imax\) is updating its policy in state \(\smax\). Naturally, \(\bar{i}^{(t)} = \bar{i}^{(t^\ast)}\) and \(\bar{s}^{(t)} = \bar{s}^{(t^\ast)}\). Consequently, for every \(a\in A_{\bar{i}^{(t)}}\),
\begin{align*}
    \pi_{\bar{i}^{(t)}}^{(t)}(a|\bar{s}^{(t)}) = \text{BR}_{\bar{i}^{(t)}}^{(t^\ast)}(a|\bar{s}^{(t)}) = \frac{\exp(\tilde{Q}_{\imax}^{(t^\ast)}(\smax,a))}{\sum_{a'\in A_{\imax}}\exp(\tilde{Q}_{\imax}^{(t^\ast)}(\smax,a'))}.
\end{align*}
Consequently, for every \(a\in A_{\imax}\),
\begin{align*}
    \pi_{\imax}^{(t)}(a|\smax) 
        &\geq \frac{\exp(\tilde{Q}_{\imax}^{(t^\ast)}(\smax,\underline{a})/\tau)}{|A_{\bar{i}^{(t)}}|\exp(\tilde{Q}_{\imax}^{(t^\ast)}(\smax,\bar{a})/\tau)} \\
        &=\frac{1}{|A_{\imax}|}\exp\left(\left(\tilde{Q}_{\imax}^{(t^\ast)}(\smax,\underline{a})-\tilde{Q}_{\imax}^{(t^\ast)}(\smax,\bar{a})\right)/\tau\right),  
\end{align*}
with \(\bar{a} \in 
\underset{{a\in A_{\imax}}}{\arg\max} \tilde{Q}_{\imax}^{(t^\ast)}(\smax,a)\) and \(\underline{a} \in \underset{{a\in A_{\imax}}}{\arg\min}  \tilde{Q}_{\imax}^{(t^\ast)}(\smax,a)\).
Since \(\pi_{\imax}^{(t)}(a|\bar{s}^{(t)})\leq 1\), it follows that for every \(a\in A_{\imax}\), 
\begin{align*}
    |\log(\pi_{\imax}^{(t)}(a|\smax) )| 
    &\leq \log(|A_{\imax}|) + \frac{1}{\tau}\left(\tilde{Q}_{\imax}^{(t^\ast)}(\smax,\bar{a})-\tilde{Q}_{\imax}^{(t^\ast)}(\smax,\underline{a})\right)\\
    & \leq \log(|A_{\imax}|) + \frac{2}{\tau}\|\tilde \bQ_{\imax}^{(t^\ast)}(\smax)\|_{\infty}.  \quad\hfill 
\end{align*}
}
\begin{lemma}\label{lem: Q_tildeBound}
    For any \(t\in [T], i\in \playerSet, s\in S\), it holds that
       $ \|\tilde{\bQ}_i^{(t)}(s)\|_{\infty} \leq C\frac{1+\tau N \log(\bar{A})}{1-\gamma},$
    where \(C \coloneqq \max_{i\in \playerSet} \|u_i\|_{\infty}\).
\end{lemma}
\proof{First, we note that for any $s\in S,$ $\pi \in \Pi$,
\begin{align*}
|\tilde{V}_i(s,\pi)| 
   &\leq \avg_\pi\ls{\sum_{k=0}^{\infty}\gamma^k|{u}_i(s^k,a^k) - \tau \sum_{j\in I_N}\nu_j(s^k,\pi_j)|}  \\ 
   &\leq \avg_\pi\ls{\sum_{k=0}^{\infty}\gamma^k\lr{|{u}_i(s^k,a^k)| + \tau  N \log(\bar{A}) }} \leq C\frac{1+\tau N \log(\bar{A})}{(1-\gamma)}.
\end{align*}
By \eqref{eqn: smooth Q}, we note that for every \(i\in \playerSet, s\in S, a_i\in A_i\),
    \begin{align*}
    |\tilde{Q}^{(t)}_i(s,a_i)| &\leq \underset{a_{-i}\sim\pi_{-i}}{\mathbb{E}}
        \Big[|{{u}_i(s,a_i,a_{-i})} - \tau\sum_{j\in I_N}\nu_j(s,\pi_j)|  +\discount\sum_{s'\in S}P(s'|s,a_i, a_{-i})\big|\tilde{V}_i(s',\pi)\big| \Big]
        \\ 
        &\leq C \underset{a_{-i}\sim\pi_{-i}}{\mathbb{E}} \Bigg[ (1+\tau N \log(\bar{A}))\lr{1 +\frac{\discount}{1-\gamma}  } \Bigg].
    \end{align*}
}

\paragraph{Proof of Lemma \ref{lemma:technical}.}
(1)     {
Fix $t\in [T]$. 
To ease the notation, let $\pi_*' \coloneqq\pi_{\imax}^{(t+1)}$, $\pi_* \coloneqq\pi_{\imax}^{(t)}$, $\pi_{-*} \coloneqq\pi_{-\imax}^{(t)}$, $\nu_* \coloneqq \nu_{\imax}$, $Q_*$ denote $\tilde{Q}_{\imax}^{(t)}$, $\bQ_*$ denote $\tilde{\bQ}_{\imax}^{(t)}$.
Note that by \eqref{eq:max_improv} and \eqref{eq: PolicyMaxMaxImprove},
\begin{align}
&\Delta_{\imax}^{(t)}(\smax)=\sum_{a\in A_{\imax}}\left(  \pi_*'(a|\smax) - \pi_*(a|\smax) \right) Q_*(\smax,a) + \tau \nu_{*}(\smax,\pi_*) - \tau\nu_{*}(\smax,\pi_*')\notag\\
&{\leq}  \sum_{a\in A_{\imax}}\left( \pi_*'(a|\smax) - \pi_*(a|\smax) \right) Q_*(\smax,a)+ \tau \sum_{a\in A_{\imax}} \log(\pi_*(a|\smax))\lr{\pi_*(a|\smax)-\pi_*'(a|\smax)}\notag\\
&\leq  \sum_{a\in A_{\imax}}\bigg(\Big|  \pi_*'(a|\smax) - \pi_*'(a|\smax) \Big| \cdot \Big| Q_*(\smax,a) - \tau  \log( \pi_*(a|\smax))\Big|\bigg), \label{eqn: delta part 1}
\end{align} 
where the first inequality follows from convexity of \(\nu_i(s,\cdot)\). By Cauchy-Schwarz inequality and noting that \(\underset{i\in\playerSet}{\max}|A_i|\leq \bar{A}\),
\begin{align*}
\eqref{eqn: delta part 1}&{\leq} \sqrt{\bar{A}}\max_{a\in A_{\imax}}\Big| Q_*(\smax,a) - \tau  \log(\pi_*(a|\smax))\Big| \cdot\left\|\pi_*'(\smax) - \pi_*(\smax) \right\|_2 \\ 
&\leq \sqrt{\bar{A}}\lr{\max_{a\in A_{\imax}}\Big|Q_*(\smax,a)\Big| + \max_{a\in A_{\imax}} \tau \Big| \log(\pi_*(a|\smax))\Big|}\cdot \left\|\pi_*'(\smax) - \pi_*(\smax) \right\|_2.
\end{align*}
Note that 
Lemma \ref{lem: LogPolicy} implies that there exists \(\hat{t}\leq t\) such that 
$\underset{a\in A_{\imax}}{\max} \tau \bigg| \log(\pi_*(a|\smax))\bigg| \leq  2\| \tilde{\bQ}_{\imax}^{(\hat{t})}(\smax)\|_{\infty}+ \tau \log(\bar{A}).$ 
Consequently, it follows that 
\begin{align*}
\Delta_{\imax}^{(t)}(\smax)&\leq  \sqrt{\bar{A}}\Big(\|\bQ_*(\smax)\|_{\infty} + 2\| \tilde{\bQ}_{\imax}^{(\hat{t})}(\smax)\|_{\infty}+ \tau \log(\bar{A})\Big)\cdot\left\|\pi_*'(\smax) - \pi_*(\smax) \right\|_2 \\ 
&{\leq} 4C\frac{1+\tau N \log(\bar{A})}{1-\gamma}\sqrt{\bar{A}}\left\|\pi_*'(\smax) - \pi_*(\smax) \right\|_2, 
\end{align*}
 where the last inequality follows from Lemma \ref{lem: Q_tildeBound}.  This concludes the proof for Lemma \ref{lemma:technical}  1).}\\
(2) Here, we show that 
\begin{align*}
    &\sum_{t=1}^{T-1}\| \pi_*'(\smax)- \pi_*(\smax)\|_2^2 \leq \frac{2}{\tau\bar{\mu}}\lr{\tilde{\Psi}(\mu,\pi^{(T)}) - \tilde{\Psi}(\mu,\pi^{(0)}) + \mnpg T }.
\end{align*}
To see this, note that for any \(t\in [T]\),
        \begin{align}
&\tilde{\Psi}(\mu,\pi^{(t+1)}) - \tilde{\Psi}(\mu,\pi^{(t)})=\tilde{\Psi}(\mu, \pi_*',\pi_{-*}) - \tilde{\Psi}(\mu, \pi_*,\pi_{-*})  \notag \\ 
&\stackrel{(i)}{\geq} \tilde{V}_{\imax}(\mu, \pi_*',\pi_{-*}) -  \tilde{V}_{\imax}(\mu, \pi_*,\pi_{-*}) - \mnpg \notag
\\
    &\stackrel{(ii)}{=} \frac{1}{1-\gamma} \sum_{\substack{s\in S}}d_{\mu}^{ \pi_*',\pi_{-*}}(s)\bigg(\big( \pi_*'(s) - \pi_*(s)\big)^\top {\bQ}_*(s) + \tau \nu_{*}(s, \pi_*) - \tau \nu_{*}(s, \pi_*')\bigg) - \mnpg  \notag\\
     &\stackrel{(iii)}{=} \frac{1}{1-\gamma} d_{\mu}^{ \pi_*',\pi_{-*}}(\smax)\bigg(( \pi_*'(\smax) - \pi_*(\smax)^\top )\cdot {\bQ}_*(\smax) + \tau \nu_{*}(\smax, \pi_*) - \tau \nu_{*}(\smax, \pi_*')\bigg)  - \mnpg, 
\label{eqn: tilde psi minus psi part 1}
    \end{align}
where \((i)\) follows from Lemma \ref{lem: perturbedGameNMPG}, \((ii)\) follows from Lemma \ref{lem: performanceDiffPG}, and \((iii)\) holds because \(\pi_*'(s) = \pi_*(s)\) for all \(s\neq \smax\). Next, from Algorithm \ref{alg: SequentialPureBRMaxPlayer}, note that \(\pi_*'(\smax) = \text{BR}_{\imax}^{(t)}(\smax)\). Consequently, using Lemma \ref{lem: FirstOrderOptimality}, we obtain 
\begin{align}
  & \eqref{eqn: tilde psi minus psi part 1}  \geq \frac{\tau d_{\mu}^{\pi_*',\pi_{-*}}(\smax)}{1-\gamma}   \Big(\log(\pi_*'(\smax))^\top \cdot\big(\pi_*'(\smax) -\pi_*(\smax)\big)\notag \\
   &\quad +   \nu_{*}(\smax,\pi_{*}) - \nu_{*}(\smax,\pi_{*}')\Big) -\alpha. \label{eqn: tilde psi minus psi part 2}
\end{align}
    Furthermore, using Lemma \ref{lem: strongConvexityEntropy}, we obtain 
\begin{align}
    \eqref{eqn: tilde psi minus psi part 2}&
    {\geq} \frac{\tau }{2(1-\gamma)} d_{\mu}^{\pi_*',\pi_{-*}'(\smax)}\|\pi_*'(\smax)-\pi_*(\smax)\|_2^2 - \mnpg  \notag\\&\stackrel{(a)}{\geq}  \frac{\tau  \bar{\mu}}{2}\|\pi_*'(\smax)-\pi_*(\smax)\|_2^2 - \mnpg,  \notag
\end{align}
 where \((a)\) follows from \(d_{\mu}^{\pi_i^{(t+1)},\pi_{-i}^{(t)}}(\smax) \geq (1-\gamma)\bar{\mu}\). Summing the above inequality over all \(t\in [T]\) yields:
\begin{align*}
    &\tilde{\Psi}(\mu,\pi^{(T)}) - \tilde{\Psi}(\mu,\pi^{(0)}) = \sum_{t\in [T]}\tilde{\Psi}(\mu,\pi^{(t+1)}) - \tilde{\Psi}(\mu,\pi^{(t)}) \\&\geq \frac{\tau  \bar{\mu}}{2}\sum_{t\in[T]}\|\pi_{*}'(\smax)-\pi_{*}(\smax)\|_2^2 - \mnpg T. 
\end{align*}
Finally to conclude Lemma \ref{lemma:technical} (2), note that
\begin{align*}
\sum_{t\in[T]}\|\pi_{*}'(\smax)-\pi_{*}(\smax)\|_2^2 
&\leq \frac{2}{\tau  \bar{\mu}}\left(\tilde{\Psi}(\mu,\pi^{(T)}) - \tilde{\Psi}(\mu,\pi^{(0)})  + \mnpg T \right) \\ 
&\leq \frac{2}{\tau  \bar{\mu}}\left(|\Phi(\mu,\pi^{(T)}) - \Phi(\mu,\pi^{(0)})|  + \frac{2\tau N \log(\bar{A})}{1-\gamma}+ \mnpg T \right),
\end{align*}
where the last inequality follows by noting that for any \(\pi, \pi'\in \Pi\) and any $\mu \in \mathcal{P}(S)$,
\begin{align*}
   |\tilde{\Psi}(\mu, \pi) - \tilde{\Psi}(\mu, \pi')| &\leq \left|{\Phi}(\mu,\pi) - {\Phi}(\mu,\pi')\right|
  + \tau \Bigg|\avg_{\pi}\bigg[\sum_{\substack{j\in \playerSet\\
   t\in \mathbb{N}}}\gamma^t \nu_j(s^t,\pi_j)\bigg]\Bigg| 
   + \tau \Bigg| \avg_{\pi'}\bigg[\sum_{\substack{j\in \playerSet\\t\in\mathbb{N}}  }\gamma^t \nu_j(s^t,\pi_j)\bigg] \Bigg| \\ 
   &\leq |{\Phi}(\mu,\pi) - {\Phi}(\mu,\pi')| + 2\tau \frac{ N \log(\bar{A})}{1-\gamma}.
\end{align*}

\section{Algorithms to solve semi-infinite linear programming}
\label{sec: FindNearPotential}

{In this section, we present an algorithm based on the stochastic gradient method from \cite{tadic2003randomized} to solve the semi-infinite linear programming problem \eqref{eq: LinearSemiInfiniteFormulation}. 
Denote $C \coloneqq N \underset{i\in \playerSet}{\max}\|u_i\|_\infty$ and
define
\begin{equation}\label{eqn: gi}
\begin{aligned}
    g(\phi, y; \pi, \pi') =\max\Bigg\{&  \max_{i \in I} \Big| \sum_{s',a'}(d^{s}(s',a';\pi_i,\pi_{-i}) 
    - d^{s}(s',a';\pi_i',\pi_{-i}))(\phi-u_i)(s',a')\Big| - y,\\
    &\quad \max_{s\in S, a\in A} |\phi(s,a)| - C
    \Bigg\},
\end{aligned}
\end{equation}
which ensures that constraint (C1) in \eqref{eq: LinearSemiInfiniteFormulation} can be rewritten as
$
 g(\phi, y ; \pi, \pi')  \leq 0, \forall \pi, \pi' \in \Pi.
$
Let $h: \mathbb{R}\rightarrow \mathbb{R}$ be a convex differentiable function such that 
$$h(x) =0  \text{ for all } x \leq 0, \text{ and } h(x) > 0  \text{ for all } x > 0.$$ A candidate choice of $h$ is $h(x) = (\max\{0, x\})^2.$ Finally, we consider step-size schedules $\{\eta_t\}_{t=1}^\infty$ and $\{\beta_t\}_{t=1}^\infty$ such that
\begin{equation}\label{eqn: step-size assumption}
\begin{aligned}
     &\lim_{t\rightarrow\infty} \beta_t = \infty, \sum_{t=1}^\infty {\eta_t^2 \beta_t^2} < \infty, \sum_{t=1}^\infty \eta_{t} = \infty, 
     \text{ and }  \eta_t > 0, \beta_t < \beta_{t+1}  \text{ for all } t \geq 0.
\end{aligned}
\end{equation}
Theorem 4 in \cite{tadic2003randomized} shows that with probability 1, $(y^{(t)}, \phi^{(t)})$ almost surely converges to a solution of \eqref{eq: LinearSemiInfiniteFormulation}. 

\begin{algorithm}\caption{Algorithm to solve \eqref{eq: LinearSemiInfiniteFormulation}  \cite{tadic2003randomized}}
    \begin{algorithmic}
        \STATE{\textbf{Input:} $y^{(0)} \in \mathbb{R}_+, \phi^{(0)} \in \mathbb{R}^{S \times A} $, $\{\eta_t \}_{t=1}^\infty$ and $\{ \beta_t\}_{t=1}^\infty$ satisfying \eqref{eqn: step-size assumption}.
}
    \FOR{\(t = 0,1,2,...,T-1\)}
    \STATE{Sample $\pi$, $\pi'$ in $\Pi$ from uniform distribution and calculate $g(\phi^{(t)}, y^{(t)}; \pi, \pi')$ in \eqref{eqn: gi}.}
     \STATE{Update $\phi^{(t)}$ with}
        \begin{equation}\label{eqn: phi_update}
        \begin{aligned}
            \phi^{(t+1)} = \phi^{(t)} - &\eta_{t+1} \beta_{t+1} h' \left(g \left(\phi^{(t)}, y^{(t)}; \pi, \pi'\right) \right)\cdot \nabla_\phi g\left(\phi^{(t)},  y^{(t)}; \pi, \pi'\right),
        \end{aligned}
        \end{equation}
   and update $y^{(t)}$ with
   \begin{align*}
       y^{(t+1)} = y^{(t)} - \eta_{t+1}\Big(& 1 +  \beta_{t+1} h' \left(g \left(\phi^{(t)},  y^{(t)}; \pi, \pi'\right) \right) \cdot \nabla_y g\left(\phi^{(t)},  y^{(t)}; \pi, \pi'\right) \Big). 
   \end{align*}
\ENDFOR
\end{algorithmic}\label{alg: FindPotential}
\end{algorithm}

\paragraph{State-wise potential games.}
Algorithm \ref{alg: FindPotential} iteratively updates the variables $y \in \mathbb{R}$ and $\phi \in \mathbb{R}^{
S\times A}$. However, this method may be slow as the dimension of \(\phi\) scales with $|S|\cdot |A|$. 
For MCGs, where each state is a static potential game, one can utilize the game structure to accelerate the convergence of algorithm.


For an MCG $\game_{\textsf{mcg}}$, there exists a function \(\phi^\ast:S\times A \rightarrow \R\) such that for every \(i\in \playerSet, s\in S, a_i, a_i' \in A_i, a_{-i}\in A_{-i},\)
$
     |\phi^\ast(s,a_i,a_{-i}) - \phi^\ast(s,a_i',a_{-i}) - (u_i(s,a_i,a_{-i}) - u_i(s,a_i',a_{-i}))|=0.$
Then one can input $\phi^{(0)} = \phi^{*}$ and omit the update of $\phi^{(t)}$ in \eqref{eqn: phi_update} in Algorithm \eqref{alg: FindPotential}.  
Figure \ref{fig:FindNearPotential} shows the empirical performance of Algorithm \ref{alg: FindPotential} for the Markov congestion game. Note that with the setting in Section \ref{sec: numerics},
$y^{(t)}$ converges to $0$,  which suggests  that $\mathcal{G}_{\textsf{mcg}}$ may be an MPG, at least for some model parameters.
}
     \begin{figure}
    \centering
    \includegraphics[width=0.4\textwidth]{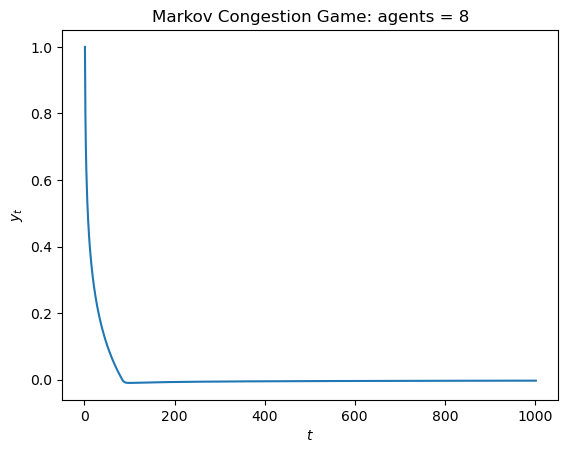}
    \caption{Find the game elasticity parameter $\mnpg$ for MCG using Algorithm \ref{alg: FindPotential}. ($\phi^{(0)} = \phi^*$, $\eta_t = \frac{1}{t}, \beta_t = t^{0.4999}, \, \forall t \geq 1.$) }
    \label{fig:FindNearPotential}
\end{figure}

\bibliographystyle{abbrv}
  \bibliography{references} 

\begin{thebibliography}{10}

\bibitem{arslan2016decentralized}
G.~Arslan and S.~Y{\"u}ksel.
\newblock Decentralized {Q}-learning for stochastic teams and games.
\newblock {\em IEEE Transactions on Automatic Control}, 62, 2016.

\bibitem{baudin2021best}
L.~Baudin and R.~Laraki.
\newblock Fictitious play and best-response dynamics in identical interest and
  zero-sum stochastic games.
\newblock In {\em International Conference on Machine Learning}, pages
  1664--1690. PMLR, 2022.

\bibitem{candogan2011flows}
O.~Candogan, I.~Menache, A.~Ozdaglar, and P.~A. Parrilo.
\newblock Flows and decompositions of games: Harmonic and potential games.
\newblock {\em Mathematics of Operations Research}, 36(3):474--503, 2011.

\bibitem{candogan2013near}
O.~Candogan, A.~Ozdaglar, and P.~A. Parrilo.
\newblock Near-potential games: Geometry and dynamics.
\newblock {\em ACM Transactions on Economics and Computation (TEAC)},
  1(2):1--32, 2013.

\bibitem{cartea2022algorithms}
{\'A}.~Cartea, P.~Chang, J.~Penalva, and H.~Waldon.
\newblock Algorithms can learn to collude: A folk theorem from learning with
  bounded rationality.
\newblock {\em Available at SSRN 4293831}, 2022.

\bibitem{cen2021fast}
S.~Cen, Y.~Wei, and Y.~Chi.
\newblock Fast policy extragradient methods for competitive games with entropy
  regularization.
\newblock {\em NeurIPS}, 34:27952--27964, 2021.

\bibitem{cui2022learning}
Q.~Cui, Z.~Xiong, M.~Fazel, and S.~S. Du.
\newblock Learning in congestion games with bandit feedback.
\newblock {\em NeurIPS}, 35:11009--11022, 2022.

\bibitem{daskalakis2020independent}
C.~Daskalakis, D.~J. Foster, and N.~Golowich.
\newblock Independent policy gradient methods for competitive reinforcement
  learning.
\newblock {\em NeurIPS}, 33:5527--5540, 2020.

\bibitem{ding2022independent}
D.~Ding, C.-Y. Wei, K.~Zhang, and M.~Jovanovic.
\newblock Independent policy gradient for large-scale {M}arkov potential games:
  Sharper rates, function approximation, and game-agnostic convergence.
\newblock In {\em International Conference on Machine Learning}, pages
  5166--5220. PMLR, 2022.

\bibitem{evans2024learning}
B.~P. Evans and S.~Ganesh.
\newblock Learning and calibrating heterogeneous bounded rational market
  behaviour with multi-agent reinforcement learning.
\newblock In {\em Proceedings of the 23rd International Conference on
  Autonomous Agents and Multiagent Systems}, 2024.

\bibitem{fox2022independent}
R.~Fox, S.~M. Mcaleer, W.~Overman, and I.~Panageas.
\newblock Independent natural policy gradient always converges in {M}arkov
  potential games.
\newblock In {\em AISTATS}, pages 4414--4425. PMLR, 2022.

\bibitem{fudenberg1998theory}
D.~Fudenberg and D.~K. Levine.
\newblock {\em The Theory of Learning in Games}, volume~2.
\newblock MIT press, 1998.

\bibitem{guo2024alpha}
X.~Guo, X.~Li, and Y.~Zhang.
\newblock An $\alpha $-potential game framework for $ n $-player dynamic games.
\newblock {\em arXiv preprint arXiv:2403.16962}, 2024.

\bibitem{hettich1993semi}
R.~Hettich and K.~O. Kortanek.
\newblock Semi-infinite programming: theory, methods, and applications.
\newblock {\em SIAM review}, 35, 1993.

\bibitem{kavuncu2021potential}
T.~Kavuncu, A.~Yaraneri, and N.~Mehr.
\newblock Potential ilqr: A potential-minimizing controller for planning
  multi-agent interactive trajectories.
\newblock In {\em 17th Robotics: Science and Systems}, 2021.

\bibitem{leonardos2021global}
S.~Leonardos, W.~Overman, I.~Panageas, and G.~Piliouras.
\newblock Global convergence of multi-agent policy gradient in {M}arkov
  potential games.
\newblock In {\em ICLR Workshop on Gamification and Multiagent Solutions},
  2022.

\bibitem{maheshwari2022independent}
C.~Maheshwari, M.~Wu, D.~Pai, and S.~Sastry.
\newblock Independent and decentralized learning in {M}arkov potential games.
\newblock {\em arXiv preprint arXiv:2205.14590}, 2022.

\bibitem{mao2021decentralized}
W.~Mao, T.~Ba{\c{s}}ar, L.~F. Yang, and K.~Zhang.
\newblock Decentralized cooperative multi-agent reinforcement learning with
  exploration.
\newblock {\em arXiv preprint arXiv:2110.05707}, 2021.

\bibitem{mertikopoulos2016learning}
P.~Mertikopoulos and W.~H. Sandholm.
\newblock Learning in games via reinforcement and regularization.
\newblock {\em Mathematics of Operations Research}, 41(4):1297--1324, 2016.

\bibitem{monderer1996potential}
D.~Monderer and L.~S. Shapley.
\newblock Potential games.
\newblock {\em Games and Economic Behavior}, 14(1):124--143, 1996.

\bibitem{narasimha2022multi}
D.~Narasimha, K.~Lee, D.~Kalathil, and S.~Shakkottai.
\newblock Multi-agent learning via {M}arkov potential games in marketplaces for
  distributed energy resources.
\newblock In {\em IEEE 61st Conference on Decision and Control (CDC)}, pages
  6350--6357, 2022.

\bibitem{papadimitriou2007complexity}
C.~H. Papadimitriou.
\newblock The complexity of finding {N}ash equilibria.
\newblock {\em Algorithmic game theory}, 2:30, 2007.

\bibitem{rosenthal1973class}
R.~W. Rosenthal.
\newblock A class of games possessing pure-strategy {N}ash equilibria.
\newblock {\em International Journal of Game Theory}, 2, 1973.

\bibitem{rudin1953principles}
W.~Rudin et~al.
\newblock {\em Principles of Mathematical Analysis}, volume~3.
\newblock McGraw-hill New York, 1976.

\bibitem{sayin2021decentralized}
M.~Sayin, K.~Zhang, D.~Leslie, T.~Basar, and A.~Ozdaglar.
\newblock Decentralized q-learning in zero-sum {M}arkov games.
\newblock {\em NeurIPS}, 34, 2021.

\bibitem{shapley1953stochastic}
L.~S. Shapley.
\newblock Stochastic games.
\newblock {\em Proceedings of the National Academy of Sciences},
  39(10):1095--1100, 1953.

\bibitem{song2021can}
Z.~Song, S.~Mei, and Y.~Bai.
\newblock When can we learn general-sum {M}arkov games with a large number of
  players sample-efficiently?
\newblock In {\em International Conference on Learning Representations}, 2021.

\bibitem{sun2023distributed}
L.~Sun, P.-Y. Hung, C.~Wang, M.~Tomizuka, and Z.~Xu.
\newblock Distributed multi-agent interaction generation with imagined
  potential games.
\newblock In {\em American Control Conference (ACC)}, pages 143--150, 2024.

\bibitem{sun2024imagined}
L.~Sun, Y.~Wang, P.-Y. Hung, C.~Wang, X.~Zhang, Z.~Xu, and M.~Tomizuka.
\newblock Imagined potential games: A framework for simulating, learning and
  evaluating interactive behaviors.
\newblock {\em arXiv preprint arXiv:2411.03669}, 2024.

\bibitem{sun2023provably}
Y.~Sun, T.~Liu, R.~Zhou, P.~Kumar, and S.~Shahrampour.
\newblock Provably fast convergence of independent natural policy gradient for
  {M}arkov potential games.
\newblock {\em NeurIPS}, 2023.

\bibitem{tadic2003randomized}
V.~B. Tadi{\'c}, S.~P. Meyn, and R.~Tempo.
\newblock Randomized algorithms for semi-infinite programming problems.
\newblock In {\em 2003 European Control Conference (ECC)}, pages 3011--3015.
  IEEE, 2003.

\bibitem{macua2018learning}
S.~Valcarcel~Macua, J.~Zazo, and S.~Zazo.
\newblock Learning parametric closed-loop policies for {M}arkov potential
  games.
\newblock In {\em 6th International Conference on Learning Representations
  (ICLR)}, 2018.

\bibitem{yongacoglu2021decentralized}
B.~Yongacoglu, G.~Arslan, and S.~Y{\"u}ksel.
\newblock Decentralized learning for optimality in stochastic dynamic teams and
  games with local control and global state information.
\newblock {\em IEEE Transactions on Automatic Control}, 67, 2021.

\bibitem{yongacoglu2023satisficing}
B.~Yongacoglu, G.~Arslan, and S.~Y{\"u}ksel.
\newblock Satisficing paths and independent multiagent reinforcement learning
  in stochastic games.
\newblock {\em SIAM Journal on Mathematics of Data Science}, 5, 2023.

\bibitem{zhang2021multi}
K.~Zhang, Z.~Yang, and T.~Ba{\c{s}}ar.
\newblock Multi-agent reinforcement learning: A selective overview of theories
  and algorithms.
\newblock {\em Handbook of reinforcement learning and control}, pages 321--384,
  2021.

\bibitem{zhangglobal}
R.~Zhang, J.~Mei, B.~Dai, D.~Schuurmans, and N.~Li.
\newblock On the global convergence rates of decentralized softmax gradient
  play in {M}arkov potential games.
\newblock In {\em NeurIPS}, 2022.

\bibitem{zhang2021gradient}
R.~C. Zhang, Z.~Ren, and N.~Li.
\newblock Gradient play in stochastic games: Stationary points and local
  geometry.
\newblock {\em IFAC-PapersOnLine}, 55(30):73--78, 2022.

\end{thebibliography}

\end{document}


\maketitle
\section{Proof of results in Section \ref{sec: NearMPG}}
\input{Appendix/ProofsSec3} 

\section{Proof of results in Section \ref{sec: Algorithms}}
\input{Appendix/ProofsSec4}

\section{Additional Experiments}
\begin{figure}
  \centering
  \begin{subfigure}{0.3\textwidth}
    \centering
    \includegraphics[width=\textwidth]{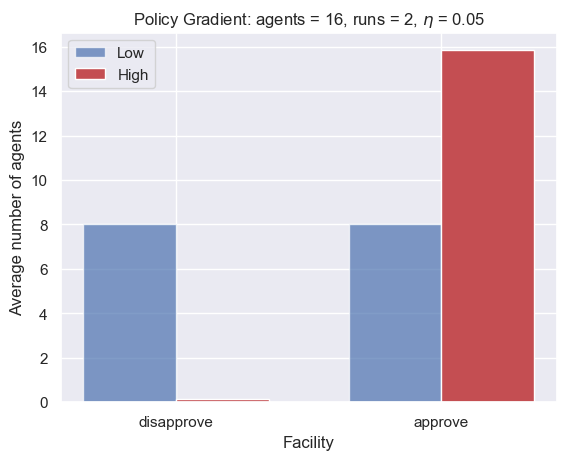}
    \caption{}
    \label{fig:a_50}
  \end{subfigure}
  \begin{subfigure}{0.3\textwidth}
    \centering
    \includegraphics[width=\textwidth]{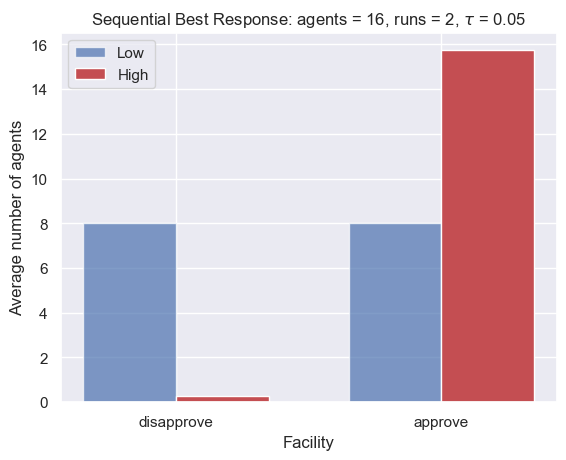}
    \caption{}
    \label{fig:b_50}
  \end{subfigure}
  \begin{subfigure}{0.33\textwidth}
    \centering \includegraphics[width=\textwidth]{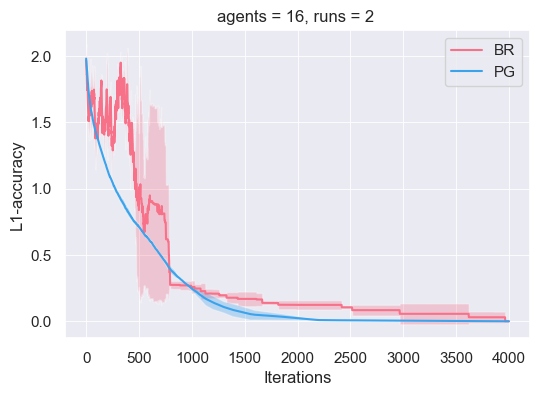}
    \caption{}
    \label{fig:c_50}
  \end{subfigure}
  \caption{\textbf{Perturbed Markov team game ($\kappa = 50$):} (a) and (b) are distributions of players taking actions in all states: (a) using policy gradient with step-size $\eta = 0.05$; (b) using sequential best response with regularizer $\tau_t =  0.9975^t \cdot  0.05 $. (c) is mean L1-accuracy with shaded region of one standard deviation over all runs. }
  \label{fig:kappa_50}
\end{figure}

In this section, we expand on the numerical experiments discussed in Section \ref{sec: numerics} and investigate the impact of random transitions on the performance of Algorithm \ref{alg: PolicyGradient} and \ref{alg: SequentialPureBRMaxPlayer}. To this end, we simulate a PMTG with the same parameters as before, i.e., $|I|=16$ agents, and $|A_i|=2$ possible actions for each agent $i\in I$. However, we introduce a stochastic transition rule based on a logistic function $
   {P} \left({\text{High}} | {\text{Low}}\right) =\frac{1}{1 + \exp \left({-\kappa \left(n(a) - \frac{|I|}{2}\right)} \right)}, 
   {P} \left({\text{High}} | {\text{High}}\right) =\frac{1}{1 + \exp \left({-\kappa \left(n(a) - \frac{|I|}{4}\right)} \right)},$
where $n(a)$ denotes the number of agents who approve the project, and $\kappa$ modulates the steepness of the transition function as it passes through its midpoint.

We set $\kappa = 50$ as the parameter in the logistic transition function. We apply a regularizer of the form $\tau_t = 0.9975^t \cdot 0.05$ in Algorithm \ref{alg: SequentialPureBRMaxPlayer} and a fixed step size $\eta = 0.05$ in Algorithm \ref{alg: PolicyGradient}. 
As shown in Figures \ref{fig:a_50} and \ref{fig:b_50}, both algorithms successfully converged to deterministic Nash policies, despite the randomness in transitions. Figure \ref{fig:c_50} further illustrates that the rate of convergence for the two algorithms is similar in this particular problem setting.

\section{Additional Contents (After first submission)}
Proposition \ref{prop: sep of value func} demonstrates that the value function of any player $i$ in a Markov $\alpha$-potential game can be decomposed into three components: the $\alpha$-Potential function $\Phi$, a term independent of player $i$'s policy, and a term bounded by $\alpha$. The first two terms are identical to the value function decomposition in MPG \citep{leonardos2021global}. The third term is unique to the Markov $\alpha$-potential game framework.

\begin{proposition}[Separability of Value Functions]\label{prop: sep of value func}
    Let \(\game = \langle \playerSet, \stateSet, (\actSet_i)_{i\in\playerSet}, (\stagePayoff_i)_{i\in\playerSet}, \transition, \discount\rangle\) be a Markov $\alpha$-potential game with $\alpha$-potential $\Phi$. Then there exists a function $U^i:S \times \Pi_{-i}\rightarrow \mathbb{R}$ and a function $r^i: S \times \Pi$ satisfying $|r^i| \leq \alpha$ such that for any  $s\in S$ and $\pi=\left(\pi_i, \pi_{-i}\right) \in \Pi$, we have $V^i(s, \pi)=\Phi(s, \pi)+U^i\left(s, \pi_{-i}\right) + r^i(s, \pi)$ for any agent $i \in I$.
\end{proposition}
\begin{proof}
Fix an admissible policy $\bar{\pi}_i \in \Pi_i$ for player $i$. Define $U(s, \pi_{-i}) := V_i(s, \bar{\pi}_i, \pi_{-i}) - \Phi(s, \bar{\pi}_i, \pi_{-i})$ and define 
\begin{equation}\label{eqn: r_i}
   r^i(s, \pi_i, \pi_{-i}) := V_i(s, \pi_i, \pi_{-i}) - \Phi(s, \pi_i, \pi_{-i}) - U^i(s, \pi_{-i}) 
\end{equation}
for any $s\in S, \pi_i\in \Pi_i, \pi_{-i} \in \Pi_{-i}.$ Plug $\bar{\pi}_i$ into \eqref{eqn: r_i} to get $ r^i(s, \bar{ \pi}_i, \pi_{-i}) = 0.$
Finally, by Definition \ref{def: eps MPG} one can obtain
    \begin{align*}
         |r^i(s, \pi_i, \pi_{-i})|  &=  \left|  r^i(s, \pi_i, \pi_{-i}) + U^i(s, \pi_{-i}) - (r^i(s, \bar{\pi}_i, \pi_{-i}) + U^i(s, \pi_{-i})) \right| \\
         &= \left|  V_i(s, \pi_i, \pi_{-i}) - \Phi(s, \pi_i, \pi_{-i}) - ( V_i(s, \bar{\pi}_i, \pi_{-i}) - \Phi(s, \bar{\pi}_i, \pi_{-i})) \right|\\
         &\leq \alpha.
    \end{align*}
\end{proof}

\bibliographystyle{siamplain}
\bibliography{references}